\documentclass[a4paper,onecolumn]{article}
\pdfoutput=1
\usepackage{amsmath}
\usepackage{amsthm}
\usepackage{nicematrix}
\usepackage{booktabs}
\usepackage{siunitx}
\usepackage{multirow}
\usepackage{array}
\usepackage{siunitx}
\usepackage{booktabs}
\usepackage{makecell}
\usepackage{amssymb}
\usepackage{algorithm}  
\usepackage{algorithmicx}  
\usepackage{algpseudocode}  
\usepackage{amsfonts,amssymb} 
\floatname{algorithm}{{Algorithm}}

\newtheorem{theorem}{Theorem}

\newtheorem{lemma}{Lemma}
\newtheorem{proposition}{Proposition}

\newtheorem{assumption}{Assumption}

\usepackage{abstract}

\usepackage{biblatex}
\title{\vspace{-3.5cm}\textbf{A Variational Inference method for Bayesian variable selection  }}
\author{Lin Guoqiang\footnote{ 
School of Mathematical Sciences, Zhejiang University, China.
Email:lgq0616@zju.edu.cn}}
\date{November 2022}

\begin{document}
    \maketitle
    \begin{abstract}
        Variable selection is a classic problem in statistics. In this paper, we consider a Bayes variable selection problem based on spike-and-slab prior with mixed normal distribution proposed by Ročková and George (2014). Motivated by Ormerod and You (2017, 2023), we use the variational inference and collapsed variational inference method to solve the Bayesian problem instead of MCMC. Like Ormerod and You (2017, 2023), we also explain how the sparsity estimator is  induced, and under certain mild assumptions, we also prove the  consistent and asymptotic results.\\
        \noindent{\textbf{Keywords:}} Bayesian variable selection,
        variational inference, spike-and-slab prior
    \end{abstract}

    \section{Introduction}
        Variable selection is an important part of statistical analysis and inference, the methods of variable selection mainly include criteria method
        \cite{AIC,BIC},  penalty method\cite{lasso,SCAD} and Bayesian method\cite{2002}.
        Overviews of recent developments in variable selection
        are \cite{Spring,review1}.
        Bayesian methods have many advantages, 
        the penalty method has a corresponding Bayesian interpretation, such as Ridge regression corresponding to the normal prior of the regression coefficient, and lasso corresponding to the Double Exponential prior of the regression coefficient.A good choice of prior could induce closed form expressions for a given model. When extracting information from the posterior to identify promising subset models, the Markov chain Monte Carlo (MCMC)\cite{2000} stochastic search or SSVS (Stochastic Search Variable Selection)\cite{GM97} approach is ofen used to identify  
        high probability models.However, MCMC and SSVS are too slow when p is large.
        
        A popular scalable alternative is variational Bayesian (VB)\cite{fox2012tutorial}, which approximates the posterior by solving an optimization problem.It is computationally efficient for Bayesian inference\cite{explaining}.The mean-field\cite{mean} is usually used for the variational family.
        This sparse variational family has been employed in various settings\cite{sparse,logit1,logit2}.Furthermore,in recent researches, some theoretical properties and guarantees are proposed, which 
        shows that VB is useful in Bayesian analysis
        \cite{frequentist,yang2020alpha,yang2020alpha,zhang2022bayesian,han2019statistical,huang2016variational,carbonetto2012scalable,wang2020simple}.
        
        In this paper, we consider the Bayesian variable selection problem with the
        continuous conjugate version of the “spike-and-slab” normal mixture prior proposed by Ročková and George (2014)\cite{EMVS}.Motivated by Ormerod et al.(2017, 2023)\cite{Ormerod2017,Ormerod2022}, we use the variational inference method to solve the Bayesian problem,
        and we also get the similar properties and results 
        as Ormerod et al.(2017, 2023).Our main contributions are as follows:\\
        (i) we show how the sparsity estimator is  induced;\\
        (ii) we prove the  consistence of the estimator and variable selection 
        in our algorithm.
        
        Jin Wang(2016)\cite{UIUC2016}
        also considers the the same prior proposed by Ročková and George (2014)\cite{EMVS},and also uses the VB to solve the problem, but this paper is different from the following points:\\
        (i) Part of the prior is different. The prior of the regression coefficient in this paper is not related to the variance of the random error term,but
        \cite{UIUC2016} is the opposite.
        In addition, we treats some parameters as hyperparameters, while \cite{UIUC2016} treats them as random variable.\\
        (ii) The algorithm is different.
        We only use VB to solve the Bayesian problem, but \cite{UIUC2016} uses
        VB and EM algorithm to solve the Bayesian problem.\\
        (iii) The assumptions and some theoretical results are different. We analyze the changes of  parameters in the iterative process of the algorithm and explains the reasons for the sparsity; but \cite{UIUC2016} has no corresponding results.
        In addition, \cite{UIUC2016} assumes that the observation value $\mathbf{x}_i \in \mathbb{R}_p(1 \leq i \leq n)$ of each independent variable is a non-random variable, but some assumptions are required for the design matrix $\boldsymbol{X}$, and the consistency result of variable selection is proved under this assumption; However, we assume that the observed value $\mathbf{x}_i \in \mathbb{R}_p(1 \leq i \leq n)$ of each independent variable is a random variable and independent and identically distributed. Under this assumption, we not only proves the consistency result of variable selection, but also prove consistency results for the estimation of the regression coefficient and the variance of the random error term.
        
\section{Linear Regression Model Setup}
        Suppose there are $n$ independent and identically distributed samples $(y_i, \mathbf{x_i}), 1 \leq i \leq n$, where $y_i$ is the response variable,
        $\mathbf{x_i} \in \mathbb{R}^p $is an independent variable, let $\mathbf{y} = (y_1, y_2, \cdot, y_n)^T $, the design matrix 
        $\mathbf{X} \in \mathbb{R}^{n\times p},$ the $i$-th row element is 
        $\mathbf{x}_i^{T}$ and
        \begin{equation}
		    \mathbf{y} | \boldsymbol{\beta},\sigma^2 \sim \mathrm{N}_n(\mathbf{X}\boldsymbol{\beta},\sigma^2\mathbf{I_n})
	    \end{equation}
	    where $\boldsymbol{\beta} = (\beta_1,\beta_2,\cdots,\beta_p)^T \in \mathbb{R}^p,$
	    $\boldsymbol{\beta}$
	    is the regression coefficient,$\sigma^2 > 0$.
	    
	    We consider the spike-and-slab prior of $\boldsymbol{\beta}$ in \cite{EMVS}
	    $$
	    \pi\left(\beta_j \mid \sigma, \gamma_j\right)= \begin{cases}\mathrm{N}\left(0,  v_0\right), & \text { if } \gamma_j=0 \\ \mathrm{~N}\left(0,  v_1\right), & \text { if } \gamma_j=1\end{cases}
	    $$
	    where $\boldsymbol{\gamma}=(\gamma_1,
	    \gamma_2,\cdots,\gamma_p)^T,\gamma_j \in \{0,1\}$ is the indicator variable,When $\gamma_j = 1$, it means that the $j$th variable should be included in the model (1).
	    And
	    $v_1 > v_0 > 0$,
	    This prior is different from Ormerod et al.( 2017), in Ormerod et al.( 2017)\cite{Ormerod2017}, 
	    $\boldsymbol{\beta}$'s distribution is
	    not related with $\boldsymbol{\gamma}$.
	    For $\boldsymbol{\gamma}$, this paper assumes that $\gamma_j(1\leq j \leq p)$ are independent and all obey the binomial distribution with mean $\rho$. Without any other prior information, we take $\rho= 0.5$. For $\sigma^2$, we consider its conjugate distribution: Inverse gamma distribution, that is, $\sigma^2 \sim \mathrm{IG}(A,B)$, same as\cite{GM97}, we take $A = B = 0.5$. Therefore, the hierarchical model in this paper can be expressed as
	    	\begin{equation}
		\begin{aligned}
			&\mathrm{\pi}(\mathbf{y} | \boldsymbol{\beta},\sigma^2) \sim \mathrm{N}_n(\mathbf{X}\boldsymbol{\beta},\sigma^2\mathbf{I_n}),\quad
			\mathrm{\pi}(\sigma^2) \sim \mathrm{IG}(A,B)\\
			&\mathrm{\pi}(\boldsymbol{\beta}|\boldsymbol{\gamma},v_0,v_1) \sim \mathrm{N}_{p}(0,\mathbf{C}_{
			\boldsymbol{\gamma}}),
			\quad \mathrm{\pi}(\boldsymbol{\gamma}|\rho) \propto 
			\rho^{\sum_{j = 1}^{p}\gamma_j}\cdot(1 - \rho)^{p-\sum_{j = 1}^{p}\gamma_j}
		\end{aligned}	
	\end{equation}
	where $\mathbf{C}_{\boldsymbol{\gamma}} = \operatorname{diag}(c_1,c_2,\ldots,c_p)$,
	$c_j = (1 - \gamma_j)v_0 + \gamma_j v_1,1\leq j\leq p $,and $\rho \in (0,1),A > 0 ,B > 0,v_1 > v_0 > 0$
	are all hyperparameters. So
	the joint distribution is 
	\begin{equation}
		\begin{aligned}
			&\mathrm{\pi}\left(\boldsymbol{\beta},\sigma^2,\boldsymbol{\gamma},\mathbf{y}\right)\\
			\propto&\mathrm{\pi}\left(\mathbf{y}|\boldsymbol{\beta},\sigma^2\right) \cdot \mathrm{\pi}\left(\boldsymbol{\beta}|\boldsymbol{\gamma},v_0,v_1\right) \cdot \mathrm{\pi}\left(\sigma^2\right)
			\cdot \pi\left(\gamma | \rho\right)\\
			\propto&{\left(\sigma^2\right)}^{-\frac{n}{2}}\cdot
			\exp\left\{-\frac{1}{2 \sigma^2}\left\|\mathbf{y}-\mathbf{X}  \boldsymbol{\beta}\right\|^2\right\}\\
			&\cdot
			\prod_{j=1}^{p}{\left[\left(1-\gamma_j\right)v_0+\gamma_{j}v_1\right]}^{-\frac{1}{2}} \cdot\exp\left\{-\frac{1}{2} \boldsymbol{\beta}^T\mathbf{C}_{
			\boldsymbol{\gamma}}^{-1}\boldsymbol{\beta}
			\right\}\\
			&\cdot 
			\rho^{\sum_{j = 1}^{p}\gamma_j}\cdot\left(1 - \rho\right)^{p-\sum_{j = 1}^{p}\gamma_j}
			\cdot {\left(\sigma^2\right)}^{-A-1}\cdot
			\exp\left\{-\frac{B}{\sigma^2} \right\}
		\end{aligned}	
	\end{equation}
	where $\|\cdot\|$ is the vector's $2$-norm.Henceforth
	\begin{equation}
		\begin{aligned}
			\ln\mathrm{\pi}(\mathbf{y},\boldsymbol{\beta},\sigma^2,\boldsymbol{\gamma}) 
			=&\left(-A-\frac{n}{2}-1\right)\ln\sigma^2-
			\frac{1}{2 \sigma^2}\left\|\mathbf{y}-\mathbf{X}  \boldsymbol{\beta}\right\|^2\\
			&-\frac{1}{2}\sum_{j=1}^{p}\ln[\left(1-\gamma_j\right)v_0+\gamma_{j}v_1] -\frac{1}{2} \boldsymbol{\beta}^T\mathbf{C}_{
			\boldsymbol{\gamma}}^{-1}\boldsymbol{\beta}\\
			&-\frac{B}{\sigma^2} + \sum_{j=1}^{p}\gamma_j\ln\left(\frac{\rho}{1-\rho}\right) + \mathrm{constant}
		\end{aligned}	
	\end{equation}
	Next, we consider the variational inference method to solve this problem.
	The choices for the factorization of
	$q\left(\boldsymbol{\beta}, \sigma^2, \boldsymbol{\gamma}\right)$is
	$q(\boldsymbol{\beta}) q\left(\sigma^2\right) \prod_{j=1}^{p} q\left(\gamma_j\right).$\\
	The density $q(\boldsymbol{\beta})$ is given by:
	\begin{equation}
	\begin{aligned}
		q(\boldsymbol{\beta}) & \propto \exp \left\{
		\mathbb{E}_{-q(\boldsymbol{\beta})}\left[
		-\frac{1}{2 \sigma^2}\|\mathbf{y}-\mathbf{X}  \boldsymbol{\beta}\|^2-\frac{1}{2} \boldsymbol{\beta}^T\mathbf{C}_{
		\boldsymbol{\gamma}}^{-1}\boldsymbol{\beta}\right]\right\} \\
		& \propto \exp \left\{-\frac{1}{2} \boldsymbol{\beta}^T\left(
		\tau\mathbf{X}^T\mathbf{X}+\frac{1}{v_0}\mathbf{I_p} + \left(\frac{1}{v_1}-\frac{1}{v_0}\right)\mathbf{W}\right) \boldsymbol{\beta}
		+\tau\boldsymbol{\beta}^T\mathbf{X}^T \mathbf{y} \right\} \\ 
		&=\mathrm{N}(\boldsymbol{\mu}, \boldsymbol{\Sigma})
	\end{aligned}
	\end{equation}
	where $\boldsymbol{\Sigma} = \left(
	\tau\mathbf{X}^T\mathbf{X}+\mathbf{D}\right)^{-1},
	\mathbf{D}=\frac{1}{v_0}\mathbf{I_p} + \left(\frac{1}{v_1}-\frac{1}{v_0}\right)\mathbf{W},
	\boldsymbol{\mu} = \tau\boldsymbol{\Sigma}\mathbf{X}^T \mathbf{y},\\
	\tau = \mathbb{E}_{q}\left[\frac{1}{\sigma^2}\right],
	\mathbf{w} =\mathbb{E}_{q}\boldsymbol{\gamma}, \mathbf{W}= \operatorname{diag}(\mathbf{w})$.\\
	The density $q\left(\sigma^2\right)$ is given by 
	\begin{equation}
	\begin{aligned}
		q\left(\sigma^2\right) \propto \exp & {\left\{
		\mathbb { E } _ { - q ( \sigma ^ { 2 } ) } \left[-\frac{n}{2} \log \left(\sigma^2\right)-\frac{1}{2 \sigma^2}\left\|\mathbf{y}-\mathbf{X\boldsymbol { \beta }}\right\|^2\right.\right.} \\
		&\left.\left.-(A+1) \log \left(\sigma^2\right)-\frac{B}{\sigma^2}\right]\right\}
	\end{aligned}
	\end{equation}
	Henceforth $q\left(\sigma^2\right) = \mathrm{IG}(A_1,B_1)$ where $A_1 = A + \frac{n}{2}$,
	$$
	\begin{aligned}
		B_1 &=B + \mathbb { E } _ { - q ( \sigma ^ { 2 } ) }\left[
		\frac{1}{2}\left\|\mathbf{y}-\mathbf{X\boldsymbol { \beta }}\right\|^2
		\right]\\
		&= B+\frac{1}{2}\left[\left\|\mathbf{y}-\mathbf{X\boldsymbol { \boldsymbol{\mu} }}\right\|^2+\operatorname{tr}\left(\mathbf{X}^T \mathbf{X} \boldsymbol{\Sigma}\right)\right].\\
	\end{aligned}
	$$
	So $\tau = \mathbb{E}_{q(\sigma^2)}\left[\frac{1}{\sigma^2}\right] = \frac{A_1}{B_1}.$\\
	Let $\mathrm{expit}(x)=\exp(x)/\left(1+\exp(x)\right)$,$\forall 1\leq j \leq p,$
	The density $q\left(\gamma_j\right)$ is given by
	\begin{equation}
	\begin{aligned}
		q(\gamma_{j}) &\propto \exp \left\{\gamma_j \mathbb{E}_{-q_{(\gamma_j)}}\left[ 
		\frac{1}{2}\ln\frac{v_0}{v_1}+\ln\frac{\rho}{1-\rho}+
		\frac{1}{2}\beta_j^2\left(\frac{1}{v_0}-\frac{1}{v_1}\right)\right]\right\}\\
		&\propto \exp \left\{\gamma_j \left[ 
		\frac{1}{2}\ln\frac{v_0}{v_1}+\ln\frac{\rho}{1-\rho}+
		\frac{1}{2}\left(\mu_j^2+\Sigma_{jj}\right)\left(\frac{1}{v_0}-\frac{1}{v_1}\right)\right]\right\}
	\end{aligned}
	\end{equation}
	So $q(\gamma_{j})=\mathrm{B}(1,w_j)$ where 
	$w_j=
	\mathrm{expit}(\eta_j),$ and 
	$$
	\eta_j = 
	\frac{1}{2}\ln\frac{v_0}{v_1}+\ln\frac{\rho}{1-\rho}+
	\frac{1}{2}\left(\mu_j^2+\Sigma_{jj}\right)\left(\frac{1}{v_0}-\frac{1}{v_1}\right).
	$$
	Combining (5)-(7), using the coordinate descent method, the specific algorithm of the hierarchical Bayes model (2) can be seen in Algorithm 1.
	\begin{algorithm}[H]
		\caption{linear regression Bayes estimation}
		\label{alg1}
		\begin{algorithmic}[1]
			\Require $(\mathbf{X},\mathbf{y},v_0,v_1,A,B,\rho)$where$
			\mathbf{X}\in 
			\mathbb{R}^{n\times p},\mathbf{y}\in\mathbb{R}^n
			,v_0 > 0,v_1 > v_0,A > 0,B > 0,\rho \in (0,1)
			$ \\
			\textbf{initial parameter value:}$\mathbf{w}^{(0)}
			=\frac{1}{2}\mathbf{1}_p,
			\tau^{(0)} = 1 $\\
			t = 1\quad;\quad $\lambda = \ln\frac{\rho}{1-\rho}$
			\While{not converge}
			\State $\mathbf{W}^{(t)}
			=\mathrm{diag}\left(\mathbf{w}^{(t-1)}\right)\quad;\quad
			\mathbf{D}^{(t)}=\frac{1}{v_0}\mathbf{I_p} + \left(\frac{1}{v_1}-\frac{1}{v_0}\right)\mathbf{W}^{(t)}$
			\State$
			\boldsymbol{\Sigma}^{(t)} = \left[
			\tau^{(t-1)}\mathbf{X}^T\mathbf{X}+\mathbf{D}^{(t)}\right]^{-1}\quad;\quad
			\boldsymbol{\mu}^{(t)} = \tau^{(t-1)}\boldsymbol{\Sigma}^{(t)}\mathbf{X}^T \mathbf{y}$
			\State $B_1
			=B+\frac{1}{2}\left[\left\|\mathbf{y}-\mathbf{X}\boldsymbol { \boldsymbol{\mu} }^{(t)}\right\|^2+\operatorname{tr}\left(\mathbf{X}^T \mathbf{X} \boldsymbol{\Sigma}^{(t)}\right)\right]
			$
			\State $\tau^{(t)} = (A + n/2)/B_{1}$
			\For{$j = 1,2,...,p$}
			\State $\eta_j^{(t)}=
			\lambda + \frac{1}{2}\ln\frac{v_0}{v_1}+
			\frac{1}{2}\left[\left(\mu_j^{(t)}\right)^2+\Sigma_{jj}^{(t)}\right]\left(\frac{1}{v_0}-\frac{1}{v_1}\right)$
			\State $w_j^{(t)}=
			\mathrm{expit}\left(\eta_j^{(t)}\right)$
			\EndFor
			\State$t = t + 1$
			\EndWhile 
			\Ensure $(\boldsymbol{\mu},\mathbf{w},\tau)$where $\boldsymbol{\mu}\in\mathbb{R}^p,	\mathbf{w}\in\mathbb{R}^p,\tau > 0$ 
		\end{algorithmic}
	\end{algorithm}
	\section{Theoretical Results}
	In this section, we will analyze the limit properties of each parameter estimation from one iteration to the next in Algorithm 1. The proof is shown in Appendix $1.2$.
	\begin{proposition}
		Given observations $\mathbf{X}, \mathbf{y}$, if $ p \leq n, \mathbf{X}$ is full Rank, then in the iterative process of algorithm $1$
		$$\lim\limits_{v_0 \to 0}\lim\limits_{w_j^{(t-1)} \to 0}\mu_j^{(t)}=0. \quad(j = 1,2,\cdots,p)$$
	\end{proposition}
	
	\begin{proposition}
		Given observations $\mathbf{X},\mathbf{y},$ if $p \leq n,\mathbf{X}$ is full Rank, and $v_0$ is fixed, then $ \forall j = 1,2,\cdots,p$,if $w_j^{(t-1)} << 1$,
		then in the iterative process of algorithm $1$
		$$w_j^{(t)}
		\leq \mathrm{expit}\left[M_{j} + \frac{1}{2}\ln\frac{v_0}{v_1}  + \mathrm{O}(w_{j}^{(t-1)})\right]
		$$
		where $M_{j}$ is a constant independent of $v_1$.
	\end{proposition}
	\begin{lemma}
		$\left(\text{Proof is in Appendix} 1.1\right)$
		if $x > 0$,we have the quantities $\mathrm{expit}(-x) = \exp(-x)+\mathrm{O}(\exp(-2x))$ and 
		$\mathrm{expit}(x) = 1 - \exp(-x)+\mathrm{O}(\exp(-2x)).$
	\end{lemma}
	\noindent{}\textbf{Remark:}
	From Proposition 1, we can see that in the iterative process, when $w_j^{(t-1)}$ is close to $0$, by selecting a suitable $v_0$ and making $v_0$ small enough, the estimated value of $\boldsymbol{\beta}_j$ in the next iteration process $\mu_j^{(t)}\approx0$.Similarly,it can be known from Proposition 2 and Lemma 1 that when $w_j^{(t-1)}$ takes a small value, then
	$$
		\mathrm{expit}\left[M_{j} + \frac{1}{2}\ln\frac{v_0}{v_1}  + \mathrm{O}(w_j^{(t-1)})\right]
		\approx \exp(M_{j} + \frac{1}{2}\ln\frac{v_0}{v_1})
	$$
	When $v_0$ is fixed and $v_1$ is large enough, it can be seen that the above formula is approximately $0$, so $w_j^{(t)}\approx 0 .$ This explains how Algorithm 1 induces the sparsity of the inclusion probability $\mathbf{w}$ and regression coefficient $\boldsymbol{\beta}$.
	
	In order to establish the consistency results of this paper, we consider the same assumptions in  \cite{Ormerod2017,2014AUS,Ormerod2022}:
	\begin{assumption}
	For $1 \leq i \leq n$ the $y_i \mid \mathbf{x}_i=\mathbf{x}_i^T \boldsymbol{\beta}_0+\varepsilon_i$ where $\varepsilon_i$ and $\varepsilon_j$ are independent if $i \neq j, \mathbb{E}\left(\varepsilon_i\right)=0, \operatorname{Var}\left(\varepsilon_i\right)=\sigma_0^2$ and $0<\sigma_0^2<\infty, \boldsymbol{\beta}_0$ are the true values of $\boldsymbol{\beta}$ and $\sigma^2$ with $\boldsymbol{\beta}_0$ being element-wise finite;
	\end{assumption}
    \begin{assumption}
    for $1 \leq i \leq n$ the random variables $\mathbf{x}_i \in \mathbb{R}^p$ are independent and identically distributed.
    \end{assumption}
    \begin{assumption}
    the $p \times p$ matrix $\mathbf{S} \equiv \mathbb{E}\left(\mathbf{x}_i \mathbf{x}_i^T\right)$ is element-wise finite and $\mathbf{X}=\left[\mathbf{X}_1, \ldots, \mathbf{X}_p\right]$ where $\operatorname{rank}(\mathbf{X})=p$;
    \end{assumption}
    \begin{assumption}
    for $1 \leq i \leq n$ the random variables $\mathbf{x}_i$ and $\varepsilon_i$ are independent.
    \end{assumption}
    The first four assumptions are the same as \cite{2014AUS,Ormerod2017,Ormerod2022}.
    \begin{assumption}
        The number of variables $p$  increases as n increases and satisfies 
    $p = O_p(n^{\frac{1}{12}}).$
    \end{assumption}
    \begin{assumption}
    $v_0$ satisfies \\
    (i)$0 < v_0 < \min\left(v_1\exp\left(-2\lambda\right),v_1\right)$;\\
    (ii)Let $S^\star = \{j:\beta_{0,j} \neq 0,1 \leq j \leq p \}$,
	$\hat{S}_n = \{j:w_{j} > 0.5,1 \leq j \leq p \},$
	if $S^\star$ is not empty set, let $l_0 =\min\{|\beta_{0,j}|: 
	j \in S^\star
	\},\delta > 0$,$v_0$ satisfies
	$$0 < v_0 < l_0^2,\quad
	\frac{l_0^2}{v_0} + \ln{v_0}
	\geq \ln v_1+\frac{l_0^2}{v_1}-2\lambda + 2\delta.
	$$
    \end{assumption}
    \noindent{}\textbf{Remark:}When $0<v_0 < l_0^2$, $f(v_0)=\frac{l_0^2}{v_0} + \ln{v_0}$ decreases monotonously with respect to $v_0$, while $v_0 \to 0$
    , $f(v_0) \to +\infty$, as long as $v_0$ takes a smaller value, the condition (ii) of assumption 6 can be satisfied, so assumption 6
    The conditions (i) and (ii) of can be satisfied at the same time and are not contradictory. We will construct the following consistency results based on the above assumptions. The proof is shown in Appendix 2.
    \begin{theorem}
    (Consistency of Coefficient Estimates) Under assumptions $1-4$, $\boldsymbol{\mu}$ and $1/\tau$ obtained by Algorithm 1 are consistent estimates of the true value of the regression coefficient $\boldsymbol{\beta}_0$ and the variance of the random error term $\sigma_0^2$.
    \end{theorem}
    \begin{theorem}
    (Consistency of Variable Selection) Under assumptions 1 - 5, we have 
    $$\operatorname{P}\left(\hat{S}_n = 
	S^\star\right) \to 1 \quad (n \to +\infty).
	$$
    \end{theorem}
    \section{Collapsed Variational Inference}
    \subsection{Model setup}
    \begin{equation}
		\begin{aligned}
			&\mathrm{\pi}(\mathbf{y} | \boldsymbol{\beta},\sigma^2) \sim \mathrm{N}_n(\mathbf{X}\boldsymbol{\beta},\sigma^2\mathbf{I_n}),\quad
			\mathrm{\pi}(\sigma^2) \sim \mathrm{IG}(A,B)\\
			&\mathrm{\pi}(\boldsymbol{\beta}|\boldsymbol{\gamma},\sigma^2 ) \sim \mathrm{N}_{p}(0,\sigma^2\mathbf{C}_{
			\boldsymbol{\gamma}}),
			\quad \mathrm{\pi}(\boldsymbol{\gamma}) \propto 
			\rho^{\sum_{j = 1}^{p}\gamma_j}\cdot(1 - \rho)^{p-\sum_{j = 1}^{p}\gamma_j}
		\end{aligned}	
	\end{equation}
    Note: The difference between this model and (2) lies in the prior variance of $\boldsymbol{\beta}$.\\
    Using the collapsed variational inference, we have
    \begin{equation}
		\begin{aligned}
			\pi(\mathbf{y}, \boldsymbol{\gamma})\propto & \int_{\boldsymbol{\beta}, \sigma^2}  \mathrm{\pi}\left(\mathbf{y}|\boldsymbol{\beta},\sigma^2\right) \cdot \mathrm{\pi}\left(\boldsymbol{\beta}|\boldsymbol{\gamma},\sigma^2\right) \cdot \mathrm{\pi}\left(\sigma^2\right)
			\cdot \pi\left(\boldsymbol{\gamma}\right)
			d\boldsymbol{\beta} d \sigma^2 \\
			\propto & \pi(\boldsymbol{\gamma})\cdot
			\prod_{j=1}^{p}{\left[\left(1-\gamma_j\right)v_0+\gamma_{j}v_1\right]}^{-\frac{1}{2}}\\
			&\cdot\int_{\sigma^2}\left(\sigma^2\right)^{-(A+(n+p) / 2+1)} \exp \left[-\frac{1}{\sigma^2}
			\left(B+\frac{\|\mathbf{y}\|^2}{2}\right) \right] \\
			& \cdot\left[\int_{\boldsymbol{\beta}} \exp \left(-\frac{1}{2} \boldsymbol{\beta}^{\top}\left[\frac{ \mathbf{X}^{\top} \mathbf{X} }{\sigma^2}+\frac{\mathbf{C}_{
			\boldsymbol{\gamma}}^{-1}}{\sigma^2 } \right] \boldsymbol{\beta}+\frac{\mathbf{y}^{\top} \mathbf{X}}{\sigma^2} \boldsymbol{\beta}\right) d \boldsymbol{\beta}\right] d \sigma^2 \\
			\propto & \pi(\boldsymbol{\gamma}) \cdot
			\prod_{j=1}^{p}{\left[\left(1-\gamma_j\right)v_0+\gamma_{j}v_1\right]}^{-\frac{1}{2}}\\
			&\cdot
			\int_{\sigma^2}\left|\left(\mathbf{X}^{\top}\mathbf{X}+ \mathbf{\mathbf{C}_{
			\boldsymbol{\gamma}}^{-1}}\right)\right|^{-1 / 2}\left(\sigma^2\right)^{-A-n / 2-1} \\
			& \cdot \exp \left[-\frac{1}{\sigma^2}\left(B+\frac{1}{2}\|\mathbf{y}\|^2-\frac{1}{2} \mathbf{y}^{\top} \mathbf{X} \left(\mathbf{X}^{\top} \mathbf{X} +\mathbf{C}_{
			\boldsymbol{\gamma}}^{-1}\right)^{-1}  \mathbf{X}^{\top} \mathbf{y}\right)\right] d \sigma^2 \\
			\propto & \pi(\boldsymbol{\gamma}) \cdot
			\prod_{j=1}^{p}{\left[\left(1-\gamma_j\right)v_0+\gamma_{j}v_1\right]}^{-\frac{1}{2}}
			\cdot \left|\left( \mathbf{X}^{\top} \mathbf{X}+\mathbf{C}_{
			\boldsymbol{\gamma}}^{-1}\right)\right|^{-1 / 2} \\
			& \cdot\left[B+\frac{1}{2}\|\mathbf{y}\|^2-\frac{1}{2} \mathbf{y}^{\top} \mathbf{X} \left( \mathbf{X}^{\top} \mathbf{X} +\mathbf{C}_{
			\boldsymbol{\gamma}}^{-1}\right)^{-1}  \mathbf{X}^{\top} \mathbf{y}\right]^{-A-n / 2} . 
		\end{aligned}
	\end{equation}
    Take $q\left( \boldsymbol{\gamma}\right)=\prod_{j=1}^{p} q\left(\gamma_j\right)$, also use the coordinate descent method, then
    \begin{equation}
		\nonumber
		\begin{aligned}
			\log q\left(\gamma_j\right) 
			=&\mathbb{E}_{-q\left(\gamma_j\right)}[\log \mathrm{\pi}(\mathbf{y}, \boldsymbol{\gamma})] + c \\
			=&\gamma_j \left(\operatorname{logit}(\rho) + \frac{1}{2}\ln\frac{v_0}{v_1}\right)
			-\frac{1}{2} \mathbb{E}_{-\mathrm{q}\left(\gamma_j\right)}\log \left| \mathbf{X}^{\top} \mathbf{X}+\mathbf{C}_{
			\boldsymbol{\gamma}}^{-1}\right| \\
			&-\left(A+\frac{n}{2}\right) \mathbb{E}_{-q\left(\gamma_j\right)} \log \left(B+\frac{1}{2}\|\mathbf{y}\|^2-\frac{1}{2} \mathbf{y}^{\top} \mathbf{X} \left(\mathbf{X}^{\top} \mathbf{X}+\mathbf{C}_{
			\boldsymbol{\gamma}}^{-1}\right)^{-1}  \mathbf{X}^{\top} \mathbf{y}\right)+ c
		\end{aligned}
	\end{equation}
    
	\begin{equation}
		\nonumber
		\begin{aligned}
			&\mathbb{E}_{-q\left(\gamma_j\right)}\log \left| \mathbf{X}^{\top} \mathbf{X}+\mathbf{C}_{
			\boldsymbol{\gamma}}^{-1}\right| \\
			& \approx\mathbb{E}_{-q\left(\gamma_j\right)}\left\{|\boldsymbol{\gamma}| \log (n)-(p-|\gamma|) \log \left(v_0\right)+\log \left|\mathbf{S}_{\gamma,\gamma}+\mathbf{O}_p^m\left(n^{-1 / 2}\right)\right|\right\} \\
			& =\left(\sum_{k \neq j} w_k+\gamma_j\right) \log (n)-\left(p-\sum_{k \neq j} w_k-\gamma_j\right) \log \left(v_0\right)+\left(\sum_{k \neq j} w_k+\gamma_j\right) O_p(1),
		\end{aligned}
	\end{equation}
    Set $k\in\{0,1\}$, then when $\gamma_{j}=k$, there is
$$
	\begin{aligned}
		& \mathbb{E}_{-q\left(\gamma_j\right)} \log \left(B+\frac{1}{2}\|\mathbf{y}\|^2-\frac{1}{2} \mathbf{y}^{\top} \mathbf{X} \left(\mathbf{X}^{\top} \mathbf{X}+\mathbf{C}_{
		\boldsymbol{\gamma}}^{-1}\right)^{-1}  \mathbf{X}^{\top} \mathbf{y}\right) \\
		& \approx \log \left(B+\frac{1}{2}\|\mathbf{y}\|^2-\frac{1}{2} \mathbf{y}^{\top} \mathbf{X} \left(\mathbf{X}^{\top} \mathbf{X}+
		\frac{1}{v_0}\mathbf{I}+\left(\frac{1}{v_1}-\frac{1}{v_0}\right)\mathbf{W}^{(jk)}\right)^{-1}  \mathbf{X}^{\top} \mathbf{y}\right)
	\end{aligned}
	$$
    Where $\mathbf{W}^{(jk)}$ represents the matrix obtained by replacing the $j$th diagonal element of the diagonal matrix $\mathbf{W}$ with $k$.
    So
    \begin{equation}
		\begin{aligned}
			q\left(\gamma_j=k\right) \propto&
			\exp \left[k \operatorname{logit}(\rho)+\frac{k}{2}\log\frac{v_0}{v_1}
			-\frac{k}{2} \log (n)-\frac{k}{2}\log \left(v_0\right)\right. \\
			& \left.+kO_p(1)-\left(A+\frac{n}{2}\right) \log \left(B+\frac{n}{2} \hat{\sigma}_{j k}^2\right)\right] \\
			\propto& \exp \left[-B_{jk}+k \alpha+O_p(1)\right]
		\end{aligned}
	\end{equation}
    where
    $B_{jk}=\left(A+\frac{n}{2}\right) \log \left(B+\frac{n}{2} \hat{\sigma}_{j k}^2\right),\hat{\sigma}_{j k}^2=\frac{1}{n}\left[\|\mathbf{y}\|^2-\mathbf{y}^{\top} \mathbf{X} \boldsymbol{\mu}^{(jk)}\right]
	, \boldsymbol{\mu}^{(jk)}=\left(\mathbf{X}^{\top} \mathbf{X}+\frac{1}{v_0}\mathbf{I}+\left(\frac{1}{v_1}-\frac{1}{v_0}\right)\mathbf{W}^{(jk)}\right)^{-1} \mathbf{X}^{\top} \mathbf{y},$
	$\alpha = \operatorname{logit}(\rho) -\frac{1}{2} \log n - \frac{1}{2}\log v_1, k\in\{0,1\}.$
    Therefore
    $$
	w_j \approx\left[1+\exp \left(- \mathrm{B}_{j 0}+\mathrm{B}_{j 1}-\alpha \right)\right]^{-1} .
	$$
    \begin{algorithm}[H]
		\caption{Collapsed Variational Inference}
		\label{algCVB}
		\begin{algorithmic}[1]
			\Require $(\mathbf{X},\mathbf{y},v_0,v_1,A,B,\rho)$where$
			\mathbf{X}\in 
			\mathbb{R}^{n\times p},\mathbf{y}\in\mathbb{R}^n
			,v_0 > 0,v_1 > v_0,A > 0,B > 0,\rho \in (0,1)
			$ \\
		    \textbf{initial value:}$\mathbf{w}^{(0)}
			=\frac{1}{2}\mathbf{1}_p$\\
			t = 1\quad;\quad $\alpha = \operatorname{logit}(\rho) -\frac{1}{2} \log n - \frac{1}{2}\log v_1$
			\While{The convergence condition is not met}
			\State $\mathbf{W}
			=\mathrm{diag}\left(\mathbf{w}^{(t-1)}\right)$
			\For{$j = 1,2,...,p$}
			\State $\boldsymbol{\mu}^{(j0)}=\left(\mathbf{X}^{\top} \mathbf{X}+\frac{1}{v_0}\mathbf{I}+\left(\frac{1}{v_1}-\frac{1}{v_0}\right)\mathbf{W}^{(j0)}\right)^{-1} \mathbf{X}^{\top} \mathbf{y}$
			\State 
			$T_{j0} = -\left(A+\frac{n}{2}\right) \log \left(B+\frac{1}{2}\|\mathbf{y}\|^2-\frac{1}{2} \mathbf{y}^{\top} \mathbf{X} \boldsymbol{\mu}^{(j 0)}\right)$
			\State $\boldsymbol{\mu}^{(j1)}=\left(\mathbf{X}^{\top} \mathbf{X}+\frac{1}{v_0}\mathbf{I}+\left(\frac{1}{v_1}-\frac{1}{v_0}\right)\mathbf{W}^{(j1)}\right)^{-1} \mathbf{X}^{\top} \mathbf{y}$
			\State 
			$T_{j1} = -\left(A+\frac{n}{2}\right) \log \left(B+\frac{1}{2}\|\mathbf{y}\|^2-\frac{1}{2} \mathbf{y}^{\top} \mathbf{X} \boldsymbol{\mu}^{(j1)}\right) + \alpha$
			\State $w_j^{(t)}=
			1 /\left(1+\exp \left(T_{j 0}-T_{j 1}\right)\right)$
			\EndFor
			\State$t = t + 1$
			\EndWhile 
			\Ensure $\mathbf{w}\in\mathbb{R}^p$ 
		\end{algorithmic}
	\end{algorithm}
    
    \begin{theorem}
    $\left(\text{Variable selection consistency}\right)$ Under the assumption of $1-5$, $v_0 = O(n^{\frac{1}{2}})$, $\mathbf{w}$ obtained by algorithm $2$ satisfies $$w_j= \begin{cases}1 -\left(1+\exp \left|O_p(n^{\frac{1}{2}})\right|\right)^{-1}, & j \in \boldsymbol{\gamma}_0, \\ O_p\left(n^{-1}\right), & j \notin \boldsymbol{\gamma}_0.\end{cases}$$ where $\boldsymbol{\gamma}_0$ is all non-zero regression A collection of coefficient subscripts.
    \end{theorem}   
    The proof is in appendix 3.

	\section{Quantile Regression Model Setup}
	In this section, we will consider the Bayes quantile regression, 
	For a given quantile $\tau$, quantile regression is often represented by the following model
	\begin{equation}
		Q_{y_i}(\tau)=\boldsymbol{x_i}^T \boldsymbol{\beta}, \quad i=1, \ldots, n
	\end{equation}
	Where the observed data $\left(\mathbf{x_i},y_i\right)$ is independent, $\mathbf{x_i} =\left(x_{i1},\ldots,x_{ip}\right)^\mathrm{T } \in \mathbb{R}^p$ is the independent variable, $y_i$ is the dependent variable, $\boldsymbol{\beta} =\left(\beta_1,\ldots,\beta_p\right)^T$ is $p $dimensional parameter.
	The parameter $\boldsymbol{\beta}$ can be obtained by minimizing the following objective function:
	\begin{equation}
		\sum_{i=1}^n \rho_\tau\left(y_i - \boldsymbol{x}_i^{\mathrm{T}} \boldsymbol{\beta}\right)
	\end{equation}
	where $\rho_\tau(y) = y(\tau - I\{y < 0\})$ is the test function, $I\{\cdot\}$ is the indicative function.\\
	In the Bayes problem, the model is usually expressed as $y_i = \boldsymbol{x_i}^{\mathrm{T}} \boldsymbol{\beta} + \epsilon_i$, where
    $\epsilon_i$ obeys the asymmetric Laplace distribution (asymmetric Laplace distribution ALD), the ALD distribution is determined by three parameters, and its probability density function is as follows
    \begin{equation}
		\pi(y|\mu,\sigma,\tau) \propto \exp\{-\frac{1}{\sigma}\rho_\tau(y - \mu)\}
	\end{equation}
	Based on this distribution, the probability density function of $y_i$ is as follows
	\begin{equation}
		\pi(y_i|x_i,b,\beta,\sigma,\tau) \propto \exp\{-\frac{1}{\sigma}\rho_\tau(y_i - x_i^T \beta)\}
	\end{equation}
	Through simple calculations, it can be found that maximizing the likelihood function of $(y_1,...,y_n)$ is equivalent to minimizing the objective function (2).
	When $y \sim\mathrm{ALD}(\boldsymbol{\mu},\sigma,\tau)$, $y$ has the following hierarchical representation
	\begin{equation}
			\begin{aligned}
				&y|e,\mu,\sigma,\tau \sim \mathrm{N}\left(\mu + c_1e,c_2\sigma e\right)\\
				&e | \sigma \sim \mathrm{Exp}\left(\frac{1}{\sigma}
				\right)
			\end{aligned}
	\end{equation}
	where $c_1=\frac{1-2 \tau}{\tau(1-\tau)}, c_2=\frac{2}{\tau(1-\tau)}$.Now consider the Bayes problem under $(12)$ based on the spike-and-slab prior, each prior setting is similar to the model $(2)$, then its Bayes layered model is
	\begin{equation}
		\begin{aligned}
			&\pi(y_i |e_i,\boldsymbol{\beta}  ,\sigma, \tau) \propto  \frac{1}{\sqrt{c_{2} \sigma e_{i}}} \exp \left\{-\frac{\left(y_i - x_i^T\boldsymbol{\beta}-c_{1} e_{i}\right)^2}{2c_{2} \sigma e_{i}}\right\}, \\
			&\pi(e_{i} | \sigma)  \propto \frac{1}{\sigma}\exp\left\{-\frac{e_{i}}{\sigma}\right\}, 
			\quad i =1, 2,\ldots, n,\\
			&\mathrm{\pi}(\boldsymbol{\beta}|\boldsymbol{\gamma},v_0,v_1) \sim \mathrm{N}_{p}(0,\mathbf{C}_{
				\boldsymbol{\gamma}}),
			\quad \pi(\sigma |v,\lambda) \propto \sigma^ {-A-1} \exp\left\{-\frac{B}{\sigma}\right\},\\
			&\pi(\gamma | \rho)	\propto \rho^{\sum_{j = 1}^{p}\gamma_j}
			\left(1 - \rho\right)^{p-\sum_{j = 1}^{p}\gamma_j}.\\		
		\end{aligned}	
	\end{equation}
	where $\mathbf{C}_{\boldsymbol{\gamma}} = \operatorname{diag}(c_1,c_2,\ldots,c_p)$,
	$c_j = (1 - \gamma_j)v_0 + \gamma_j v_1,1\leq j\leq p $,and $\rho \in (0,1),A > 0 ,B > 0,v_1 > v_0 > 0$ are hyperparameters
	.Let $\mathbf{e} = \left(e_1,e_2,\cdots,e_n\right)^T,$
	Henceforth
	\begin{equation}
		\nonumber
		\begin{aligned}
			& \pi(\boldsymbol{\beta},\sigma,\rho,
			\boldsymbol{\gamma},\mathbf{e} , \mathbf{y})\\
			\propto&\prod_{i=1}^{n} \frac{1}{\sqrt{c_{2} \sigma e_{i}}} \exp \{-\frac{(y_i - x_i^T\boldsymbol{\beta}-c_{1} e_{i})^2}{2c_{2} \sigma e_{i}}\}\\
			&\cdot \prod_{j=1}^{p}{\left[\left(1-\gamma_j\right)v_0+\gamma_{j}v_1\right]}^{-\frac{1}{2}}\cdot\exp\left\{-\frac{1}{2} \boldsymbol{\beta}^T\mathbf{C}_{
			\boldsymbol{\gamma}}^{-1}\boldsymbol{\beta}
			\right\}\\
			&\cdot \sigma^ {-A-1} \exp\left\{-\frac{B}{\sigma}\right\}
			\cdot \prod_{i=1}^{n}  \frac{1}{\sigma}\exp\left\{-\frac{e_{i}}{\sigma}\right\}\\
			&\cdot \rho^{\sum_{j = 1}^{p}\gamma_j}(1 - \rho)^{p-\sum_{j = 1}^{p}\gamma_j}.\\
		\end{aligned}	
	\end{equation}
	so 	\begin{equation}
		\nonumber
		\begin{aligned}
			\ln\pi(\boldsymbol{\beta},\sigma,\rho,
			\boldsymbol{\gamma},\mathbf{e} , \mathbf{y})
			=& - 
			\frac{1}{\sigma}\left[
			\sum_{i=1}^{n}\frac{(y_i- x_i^T\boldsymbol{\beta}-c_{1} e_{i})^2}{2c_{2} e_{i}} +\sum_{i=1}^{n}e_{i}
			+B\right]
			\\
			&-\frac{3n+2A+2}{2}\ln\sigma-\frac{1}{2}\sum_{j=1}^{p}\ln\left[(1-\gamma_j)v_0+\gamma_{j}v_1\right]\\
			&-\frac{1}{2} \boldsymbol{\beta}^T\mathbf{C}_{
			\boldsymbol{\gamma}}^{-1}\boldsymbol{\beta}
			+ \sum_{j = 1}^{p}\gamma_{j}\ln\left(\frac{\rho}{1 - \rho}\right) \\
			&- \frac{1}{2}\sum_{i=1}^{n}\ln e_{i}.
		\end{aligned}
	\end{equation}
	We also consider using the variational inference method to solve the problem,and
	the choices for the factorization of
	$q\left(\boldsymbol{\beta}, \sigma^2, \boldsymbol{\gamma}\right)$ is the same 
	as linear regression.\\
	The density $q(\boldsymbol{\beta})$ is given by:
	\begin{equation}
		\begin{aligned}
			q(\boldsymbol{\beta}) & \propto \exp \left\{
			\mathbb{E}_{-q(\boldsymbol{\beta})}\left[
			-\frac{1}{2 c_2\sigma}
			\left(\mathbf{y}-\mathbf{X}  \boldsymbol{\beta} - c_1\mathbf{e} \right)^T\operatorname{diag}\left(\mathbf{e}\right)^{-1}
			\left(\mathbf{y}-\mathbf{X}  \boldsymbol{\beta} - c_1\mathbf{e} \right)
			-\frac{1}{2} \boldsymbol{\beta}^T\mathbf{C}_{
			\boldsymbol{\gamma}}^{-1}\boldsymbol{\beta}\right]\right\} \\
			& \propto \exp \left\{-\frac{1}{2} \boldsymbol{\beta}^T
			\left(
			\mathbf{X}^T\mathbf{E}\mathbf{X}\tau/c_2
			+\frac{1}{v_0}\mathbf{I_p} + \left(\frac{1}{v_1}-\frac{1}{v_0}\right)\mathbf{W}\right) \boldsymbol{\beta}
			+\boldsymbol{\beta}^T\mathbf{X}^T \mathbf{y}_0\tau/c_2 \right\} \\ 
			&=\mathrm{N}(\boldsymbol{\mu}, \boldsymbol{\Sigma})
		\end{aligned}
	\end{equation}
	where $\boldsymbol{\Sigma} = \left(
	\mathbf{X}^T\mathbf{E}\mathbf{X}\tau/c_2+\mathbf{D}\right)^{-1},\mathbf{E} = \mathbb{E}_{q}\left[{\mathrm{diag}(\mathbf{e})}
	^{-1}\right],
	\mathbf{D}=\frac{1}{v_0}\mathbf{I_p} + \left(\frac{1}{v_1}-\frac{1}{v_0}\right)\mathbf{W},
	\\
	\boldsymbol{\mu} = \boldsymbol{\Sigma}\mathbf{X}^T \mathbf{y}_0
	\tau/c_2,
	\mathbf{y}_0 =
	\mathbf{E}\cdot\mathbf{y} - c_1\mathbf{1}_n,
	\tau = \mathbb{E}_{q}\left[\frac{1}{\sigma}\right],
	\mathbf{w} =\mathbb{E}_{q}\boldsymbol{\gamma}, \mathbf{W}= \operatorname{diag}(\mathbf{w})$.\\
	The density $q\left(\sigma\right)$ is given by:
	\begin{equation}
		\begin{aligned}
			q\left(\sigma\right) \propto \exp & {\left\{
				\mathbb { E } _ { - q ( \sigma ) } \left[\frac{3n+2A+2}{2}\ln\sigma-\frac{1}{\sigma}\left(\sum_{i=1}^{n}e_{i}
				+B\right)\right.\right.} \\
			&\left.\left.
			-\frac{1}{\sigma}\sum_{i=1}^{n}\frac{\left(y_i- x_i^T\boldsymbol{\beta}
			-c_{1} e_{i}\right)^2}{2c_{2} e_{i}}
			\right]\right\}
		\end{aligned}
	\end{equation}
	Henceforth $q\left(\sigma\right) = \mathrm{IG}(A_1,B_1)$,
	where $A_1 = A + \frac{3n}{2}$,
		\begin{equation}
	\begin{aligned}
		B_1 =& B + \mathbb { E } _ {  q  }\left[
		\sum_{i=1}^{n}e_{i}+\sum_{i=1}^{n}\frac{\left(y_i- x_i^T\boldsymbol{\beta}
			-c_{1} e_{i}\right)^2}{2c_{2} e_{i}}
		\right]\\
		=& B + \sum_{i=1}^{n}
		\left\{\left(1+\frac{c_1^2}{2c_2}\right)
		\mathbb { E } _ {q  }\left[e_{i}\right]
		-\frac{c_1}{c_2}\Delta_{1,i}\right. \\
		& \left. 
		+ \frac{1}{2c_2}\Delta_{2,i}
		\cdot
		\mathbb { E } _ {q  }
		\left[{e_i}^{-1}
		\right]
		\right\}
	\end{aligned}
	\end{equation}
	and 
	\begin{equation}
	\begin{gathered}
	\Delta_{1,i}=\mathbb { E } _ {q  }\left[
	y_i- x_i^T\boldsymbol{\beta}
	\right] =y_i- x_i^T\boldsymbol{\mu}  \\
	\Delta_{2,i}=\mathbb { E } _ {q  }\left[
	\left(y_i- x_i^T\boldsymbol{\beta}
	\right)^2\right] 
	= \left(y_i- x_i^T\boldsymbol{\mu}
	\right)^2+x_i^T\boldsymbol{\Sigma}x_i.
	\end{gathered}
 	\end{equation}
 	So$$\tau = \mathbb{E}_{q}\left[\frac{1}{\sigma}\right] = \frac{A_1}{B_1}.$$
 	 $\forall 1 \leq i \leq n$, the $q(e_i)$ is given by:
 	\begin{equation}
		\nonumber
		\begin{aligned}
		q(e_i) \propto &
		\exp \left\{
		\mathbb { E } _ { - q ( e_i ) }\left[-\frac{1}{\sigma}\left(e_{i}
		+\frac{\left(y_i- x_i^T\boldsymbol{\beta}-c_{1} e_{i}\right)^2}{2c_{2} e_{i}}\right)
		\right]-\frac{1}{2}\ln e_i\right\}\\
		= &
		\exp \left\{
		\mathbb { E } _ { - q ( e_i ) }\left[-\frac{1}{2\sigma}\left(
		\left(2+\frac{c_1^2}{c_2}\right)e_{i}
		+\frac{\left(y_i- x_i^T\boldsymbol{\beta}\right)^2}{c_{2}}
		\frac{1}{e_{i}}\right)
		\right]-\frac{1}{2}\ln e_i\right\}\\
		= &e_i^{\frac{1}{2} - 1}
		\exp\left\{-\frac{1}{2}\left[ 
		\frac{\tau \Delta_{2,i}}{c_2}\frac{1}{e_i}
		+
		\frac{\tau\left(2c_2+c_1^2\right)}{c_2}
		e_i\right]
		\right\}\\
		=& \mathrm{GIG}
		(\frac{1}{2},\lambda_1,\lambda_2).
		\end{aligned}
	\end{equation}
	where $\Delta_{2,i}$ is as (17),	$\lambda_1 = \frac{\tau \Delta_{2,i}}{c_2},
	\lambda_2 = \frac{\tau\left(2c_2+c_1^2\right)}{c_2},$
	$\mathrm{GIG}\left(\lambda_0,\lambda_1,\lambda_2\right)$ is the Generlized Inverse Distribution.
	If $X \sim \mathrm{GIG}\left(\lambda_0,\lambda_1,\lambda_2\right)$, then its probability density function has the form
	$$\mathrm{p}(x) \propto x^{\lambda_0 - 1} \exp\left\{-\frac{1}{2}\left(\lambda_1 x^{-1}+\lambda_2 x\right)\right\}, (x > 0)
	$$
	and the calculation formula of each order moment is
	\begin{equation}
		E[X^\alpha] = \left(\frac{\lambda_1}{\lambda_2}\right)^{\frac{\alpha}{2}}
		\frac{K_{\lambda_0 + \alpha}
			(\sqrt{\lambda_1\lambda_2}))}{K_{\lambda_0}(\sqrt{\lambda_1\lambda_2})}
	\end{equation}
	where $K_v(\cdot)$ is the third kind of modified Bessel function, satisfying
	\begin{equation}
		\begin{aligned}
			&K_v(z) = K_{-v}(z) \\
			&K_{v+1}(z) - K_{v-1}(z) = \frac{2v}{z}K_{v}(z)
		\end{aligned}
	\end{equation}
	so 
	\begin{equation}
	\frac{K_\frac{3}{2}(z)}{K_\frac{1}{2}(z)}
	= 1 + \frac{1}{z}
	\end{equation}
	Henceforth
	\begin{equation}
	\mathbb { E } _ {q  }\left[e_{i}^{-1}\right]
	=\left(\frac{\lambda_1}{\lambda_2}\right)^{-\frac{1}{2}},
	\mathbb { E } _ {q  }\left[e_{i}\right]
	=\left(\frac{\lambda_1}{\lambda_2}\right)^{\frac{1}{2}}
	\left[1+\left(\lambda_1\lambda_2 \right)^{-\frac{1}{2}}\right]
	\end{equation}
	And $\forall 1 \leq j \leq p$, the probability density of $q(\gamma_{j})$
	is given by
	\begin{equation}
		\begin{aligned}
			q(\gamma_{j}) &\propto \exp \left\{\gamma_j \mathbb{E}_{-q_{(\gamma_j)}}\left[ 
			\frac{1}{2}\ln\frac{v_0}{v_1}+\ln\frac{\rho}{1-\rho}+
			\frac{1}{2}\beta_j^2\left(\frac{1}{v_0}-\frac{1}{v_1}\right)\right]\right\}\\
			&\propto \exp \left\{\gamma_j \left[ 
			\frac{1}{2}\ln\frac{v_0}{v_1}+\ln\frac{\rho}{1-\rho}+
			\frac{1}{2}\left(\mu_j^2+\Sigma_{jj}\right)\left(\frac{1}{v_0}-\frac{1}{v_1}\right)\right]\right\}
		\end{aligned}
	\end{equation}
	So $q(\gamma_{j})=\mathrm{B}(1,w_j)$, where $w_j=
	\mathrm{expit}(\eta_j),$ and 
	$$
	\eta_j = 
	\frac{1}{2}\ln\frac{v_0}{v_1}+\ln\frac{\rho}{1-\rho}+
	\frac{1}{2}\left(\mu_j^2+\Sigma_{jj}\right)\left(\frac{1}{v_0}-\frac{1}{v_1}\right)
	$$
	Similarly,using the coordinate descent method, the specific algorithm of the  Bayes quantile model can be seen in Algorithm 2.
	\begin{algorithm}[H]
		\caption{quantile regression Bayes estimation}
		\label{alg2}
		\begin{algorithmic}[1]
			\Require $(\mathbf{X},\mathbf{y},\tau,v_0,v_1,A,B,\rho)$ where $
			\mathbf{X}\in 
			\mathbb{R}^{n\times p},\mathbf{y}\in\mathbb{R}^n,
			\tau \in (0,1),v_0 > 0,v_1 > v_0,A > 0,B > 0,\rho \in (0,1)
			$ \\
			\textbf{initial parameter value :}$\mathbf{w}^{(0)}
			=\frac{1}{2}\mathbf{1}_p,
			\tau^{(0)} = 1,\mathbf{E}_1^{(0)}= \mathbf{1}_p,\mathbf{E}_2^{(0)}= \mathbf{1}_p$\\
			t = 1, $\lambda = \ln\frac{\rho}{1-\rho},
			c_1=\frac{1-2 \tau}{\tau(1-\tau)}
			, c_2=\frac{2}{\tau(1-\tau)}$
			\While{not converge}
			\State $\mathbf{W}^{(t)}
			=\mathrm{diag}\left(\mathbf{w}^{(t-1)}\right)\quad;\quad
			\mathbf{E}^{(t)} = \mathrm{diag}
			\left(\mathbf{E}_1^{(t-1)}\right)
			$
			\State $
			\mathbf{D}^{(t)}=\frac{1}{v_0}\mathbf{I_p} + \left(\frac{1}{v_1}-\frac{1}{v_0}\right)\mathbf{W}^{(t)} \quad ;\quad
			\mathbf{y}_0 =
			\mathbf{E}^{(t)}\mathbf{y} - c_1\mathbf{1}_n$
			\State$
			\boldsymbol{\Sigma}^{(t)} = \left(
			\mathbf{X}^T\mathbf{E}^{(t)}\mathbf{X}\tau^{(t-1)}/c_2+\mathbf{D}^{(t)}\right)^{-1}\quad;\quad
			\boldsymbol{\mu}^{(t)} = \boldsymbol{\Sigma}^{(t)}\mathbf{X}^T \mathbf{y}_0\tau^{(t-1)}/c_2$
			\For{$i = 1,2,...,n$}
			\State 
			$\Delta_{1,i}=y_i- x_i^T\boldsymbol{\mu}^{(t)}$
			\State 
			$\Delta_{2,i} = \left(y_i- x_i^T\boldsymbol{\mu}
			^{(t)}\right)^2+
			x_i^T\boldsymbol{\Sigma}^{(t)}x_i$
			\EndFor
			\State $B_1
			=B + \sum_{i=1}^{n}
			\left\{\left(1+\frac{c_1^2}{2c_2}\right)
			\mathbf{E}_{2,i}^{(t-1)}
			-\frac{c_1}{c_2}\Delta_{1,i}
			+ \frac{1}{2c_2}\Delta_{2,i}
			\cdot
			\mathbf{E}_{1,i}^{(t-1)}
			\right\}
			$
			\State $\tau^{(t)} = (A + n/2)/B_{1}$
			\For{$i = 1,2,...,n$}
			\State $\lambda_1 = \frac{\tau^{(t)} \Delta_{2,i}}{c_2}\quad;\quad
			\lambda_2 = \frac{\tau^{(t)}\left(2c_2+c_1^2\right)}{c_2}$
			\State
			$
			\mathbf{E}_{1,i}^{(t)}
			=\left(\frac{\lambda_1}{\lambda_2}\right)^{-\frac{1}{2}}\quad;\quad
			\mathbf{E}_{2,i}^{(t)}
			=\left(\frac{\lambda_1}{\lambda_2}\right)^{\frac{1}{2}}
			\left[1+\left(\lambda_1\lambda_2 \right)^{-\frac{1}{2}}\right]
			$
			\EndFor
			
			\For{$j = 1,2,...,p$}
			\State $\eta_j^{(t)}=
			\lambda + \frac{1}{2}\ln\frac{v_0}{v_1}+
			\frac{1}{2}\left[\left(\mu_j^{(t)}\right)^2+\Sigma_{jj}^{(t)}\right]\left(\frac{1}{v_0}-\frac{1}{v_1}\right)$
			\State $w_j^{(t)}=
			\mathrm{expit}\left(\eta_j^{(t)}\right)$
			\EndFor
			\State$t = t + 1$
			\EndWhile 
			\Ensure $(\boldsymbol{\mu},\mathbf{w},\tau)$ where  $\boldsymbol{\mu}\in\mathbb{R}^p,
			\mathbf{w}\in\mathbb{R}^p,\tau > 0$ 
		\end{algorithmic}
	\end{algorithm}
	\section{Logistic regression}
	In this section,we will consider the 
	logistic regression.
	Consider a logistic regression model, assuming there are $n$ groups of independent observation samples
    $\left(\mathbf{x_i},y_i\right),\quad i=1, \ldots, n$, where $\mathbf{x_i} = \left(x_{i1},\ldots,x_{ip} \right)^\mathrm{T} \in \mathbb{R}^p$ is the independent variable, $y_i \in \{0,1\}$ is the dependent variable, $\boldsymbol{\beta} = \left( \beta_1,\ldots,\beta_p\right)^T$ is $p$ dimension parameter,
    Under the logistic regression model, the likelihood function is
	\begin{equation}
	p(\mathbf{y} \mid \mathbf{X},  \boldsymbol{\beta}) \propto \prod_{i=1}^n \frac{
	\exp \left(y_i\mathbf{x}_i^T \boldsymbol{\beta}\right)}{1+
	\exp \left(\mathbf{x}_i^T \boldsymbol{\beta}\right)}
	\end{equation}
	Also, consider the spike-and-slab of $\boldsymbol{\beta}$ and the prior of the prior hidden variable $\boldsymbol{\gamma}$,
	Under the Bayesian framework, the likelihood function of $(23)$ has no prior distribution. We consider the data enhancement method\cite{PG} in the literature and introduce hidden variables $\mathbf{v}$ to make the full log-likelihood of $\boldsymbol{\beta}$ be a normal distribution.
	$$
	p(\mathbf{y}, \mathbf{v} \mid \mathbf{X},   \boldsymbol{\beta})=\prod_{i=1}^n p\left(y_i \mid v_i, \mathbf{x}_i,v_i,  \boldsymbol{\beta}\right) g\left(v_i \mid m_i, 0\right)
	$$
	where
	$$
	g\left(v_i \mid m_i, 0\right) \propto
	\exp \left\{\left(y_i-\frac{1}{2}\right)\left(\mathbf{x}_i^T \boldsymbol{\beta}\right)-\frac{1}{2}\left(\mathbf{x}_i^T \boldsymbol{\beta}\right)^2\right\}
	$$
	and $g\left(v \mid b,c\right)$ is the probability density function of the Pólya-Gamma distribution.
	According to \cite{PG}, the Pólya-Gamma distribution has the following properties.\\
	if $V \sim P G(b, c)$, then its probability density function satisfies:
	$$
	g(w \mid b, c) \propto \exp \left\{-\frac{c^2}{2} w\right\} g(v \mid b, 0)
	$$
	where $g(v \mid b, 0)$ is the probability density function of $P G(b, 0)$ distribution;
	and its expectation is
	$$
	\mathbb{E}(V)=\frac{b}{2 c} \tanh \left(\frac{c}{2}\right) .
	$$
	and 
	$$
	\frac{\left(e^\psi\right)^a}{1+e^\psi}=2^{-b} e^{\kappa \psi} \int_0^{\infty} e^{-v \psi^2 / 2} g(v \mid 1, 0) \mathrm{d}v
	$$
	where $a\in \mathbb{R} ,v \sim PG(1,0),k = a - \frac{1}{2}$
	So the complete log-likelihood function of logistic regression model has the form
	\begin{equation}
		p(\mathbf{y}, \mathbf{v} \mid  \boldsymbol{\beta}) \propto \prod_{i=1}^n \exp \left\{\left(y_i-\frac{1}{2}\right)\left(\mathbf{x}_i^T \boldsymbol{\beta}\right)-\frac{v_i}{2}\left(\mathbf{x}_i^T \boldsymbol{\beta}\right)^2\right\} g\left(v_i \mid 1, 0\right)
	\end{equation}
	Therefore, $\boldsymbol{\beta}$ has the form of normal distribution, and the spike-and-slab prior is its conjugate prior. According to the formula(24), the Bayesian hierarchical model of logistic regression is
	\begin{equation}
		\begin{aligned}
		&p(\mathbf{y} \mid \mathbf{v}, \boldsymbol{\beta}) \propto \prod_{i=1}^n \exp \left\{\left(y_i-\frac{1}{2}\right)\left(\mathbf{x}_i^T \boldsymbol{\beta}\right)-\frac{v_i}{2}
		\left(\mathbf{x}_i^T \boldsymbol{\beta}\right)^2\right\}\\
		&\pi(\mathbf{v})  \propto 
		\prod_{i = 1}^{n}g(v_i\mid1,0),\quad
		\pi\left(\boldsymbol{\beta} \mid \gamma, v_0, v_1\right) \sim \mathrm{N}_p\left(0, \mathbf{C}_\gamma\right)\\
		&\pi(\boldsymbol{\gamma} \mid \rho) \propto \rho^{\sum_{j=1}^p \gamma_j} \cdot(1-\rho)^{p-\sum_{j=1}^p \gamma_j}		
		\end{aligned}
	\end{equation}
	The definitions of the coefficients are the same as in the linear regression, so
	\begin{equation}
		\nonumber
		\begin{aligned}
			& \pi(\boldsymbol{\beta},\rho,
			\boldsymbol{\gamma},\mathbf{v} , \mathbf{y})\\
			\propto&\prod_{i=1}^n \exp \left\{\left(y_i-\frac{1}{2}\right)\left(\mathbf{x}_i^T \boldsymbol{\beta}\right)-\frac{v_i}{2}
			\left(\mathbf{x}_i^T \boldsymbol{\beta}\right)^2\right\}\\
			&\cdot \prod_{j=1}^{p}{\left[\left(1-\gamma_j\right)v_0+\gamma_{j}v_1\right]}^{-\frac{1}{2}}\cdot\exp\left\{-\frac{1}{2} \boldsymbol{\beta}^T\mathbf{C}_{
			\boldsymbol{\gamma}}^{-1}\boldsymbol{\beta}
			\right\}\\
			&\cdot\prod_{i = 1}^{n}g(v_i\mid1,0)\cdot \rho^{\sum_{j = 1}^{p}\gamma_j}(1 - \rho)^{p-\sum_{j = 1}^{p}\gamma_j}.\\
		\end{aligned}	
	\end{equation}
	so
	\begin{equation}
		\nonumber
		\begin{aligned}
			\ln\pi(\boldsymbol{\beta},\rho,
			\boldsymbol{\gamma},\mathbf{v} , \mathbf{y})
			=& 
			-\frac{1}{2}\sum_{i=1}^{n}v_i\boldsymbol{\beta}^T
			\mathbf{x_i}\mathbf{x_i}^T\boldsymbol{\beta}
			+\sum_{i=1}^{n}(y_i-\frac{1}{2})
			\mathbf{x_i}^T\boldsymbol{\beta}
			\\
			&-\frac{1}{2}\sum_{j=1}^{p}\ln\left[(1-\gamma_j)v_0+\gamma_{j}v_1\right] -\frac{1}{2} \boldsymbol{\beta}^T\mathbf{C}_{
			\boldsymbol{\gamma}}^{-1}\boldsymbol{\beta}\\
			&- \sum_{i=1}^{n}\ln g(v_i|1,0)+ \sum_{j = 1}^{p}\gamma_{j}\ln\left(\frac{\rho}{1 - \rho}\right).
		\end{aligned}
	\end{equation}
	Consider solving this problem with variational inference,we can get 
	$$
	q(\boldsymbol{\beta}) \propto \mathrm{N}(\boldsymbol{\mu}, \boldsymbol{\Sigma})
	$$
	where $\boldsymbol{\Sigma} = \left(
	\mathbf{X}^T\mathbf{E}\mathbf{X}+\mathbf{D}\right)^{-1},
	\mathbf{E} = \mathbb{E}_{q}\left[{\mathrm{diag}(\mathbf{v})}
	\right],
	\mathbf{D}=\frac{1}{v_0}\mathbf{I_p} + \left(\frac{1}{v_1}-\frac{1}{v_0}\right)\mathbf{W},
	\\
	\boldsymbol{\mu} = \boldsymbol{\Sigma}\mathbf{X}^T \mathbf{y}_0
	,
	\mathbf{y}_0 =
	\mathbf{y} - \frac{1}{2}\mathbf{1}_n,
	\tau = \mathbb{E}_{q}\left[\frac{1}{\sigma}\right],
	\mathbf{w} =\mathbb{E}_{q}\boldsymbol{\gamma}, \mathbf{W}= \operatorname{diag}(\mathbf{w})$.\\
	And 
	$$ q\left(v_i\right) = PG(1,c_i)\quad
	 \forall 1 \leq i \leq n$$
	where $c_i = \mathbb{E}_{q}\left[\boldsymbol{\beta}^T\mathbf{x_i}\mathbf{x_i}^T\boldsymbol{\beta}
	\right] = \mathbf{x}_i^T\boldsymbol{\Sigma}\mathbf{x}_i
	+\boldsymbol{\mu}^T\mathbf{x_i}\mathbf{x_i}^T\boldsymbol{\mu}
	$, so $\mathbb{E}\left[v_i\right] = \frac{1}{2 c_i} \tanh \left(\frac{c_i}{2}\right)$.\\
	Similarly,$q(\gamma_{j})=\mathrm{B}(1,w_j)$,where $w_j=
	\mathrm{expit}(\eta_j),$ and 
	$$
	\eta_j = 
	\frac{1}{2}\ln\frac{v_0}{v_1}+\ln\frac{\rho}{1-\rho}+
	\frac{1}{2}\left(\mu_j^2+\Sigma_{jj}\right)\left(\frac{1}{v_0}-\frac{1}{v_1}\right)
	$$
	Therefore, the algorithm for solving the logistic model is  \textbf{algorithm 3}.
	\begin{algorithm}[H]
		\caption{logistic regression Bayes estimation}
		\begin{algorithmic}[1]
			\Require $(\mathbf{X},\mathbf{y},v_0,v_1,\rho)$ where $
			\mathbf{X}\in 
			\mathbb{R}^{n\times p},\mathbf{y}\in\mathbb{R}^n
			,v_0 > 0,v_1 > v_0,\rho \in (0,1)$ \\
			\textbf{initial parameter :}$\mathbf{w}^{(0)}
			=\frac{1}{2}\mathbf{1}_p,
			\mathbf{v}^{(0)} = \mathbf{1}_n $\\
			t = 1\quad;\quad $\lambda = \ln\frac{\rho}{1-\rho}$
			\quad;\quad $\mathbf{y}_0 =
			\mathbf{y} - \frac{1}{2}\mathbf{1}_n$
			\While{not converge}
			\State $\mathbf{W}^{(t)}
			=\mathrm{diag}\left(\mathbf{w}^{(t-1)}\right)\quad;
			\quad
			\mathbf{E}^{(t)} = \mathrm{diag}\left(\mathbf{v}^{(t-1)}\right)$
			\State $\mathbf{D}^{(t)}=\frac{1}{v_0}\mathbf{I_p} + \left(\frac{1}{v_1}-\frac{1}{v_0}\right)\mathbf{W}^{(t)}$
			\State$
			\boldsymbol{\Sigma}^{(t)} = \left[
			\mathbf{X}^T\mathbf{E}^{(t)}\mathbf{X}+\mathbf{D}^{(t)}\right]^{-1}\quad;\quad
			\boldsymbol{\mu}^{(t)} = \boldsymbol{\Sigma}^{(t)}\mathbf{X}^T \mathbf{y}_0$
			\For{$i = 1,2,...,n$}
				\State $c_i = \mathbf{x}_i^T\boldsymbol{\Sigma}^{(t)}\mathbf{x}_i
				+{\boldsymbol{\mu}^{(t)}}^T\mathbf{x_i}\mathbf{x_i}^T\boldsymbol{\mu}^{(t)}$
				\State $\mathbf{v}_i^{(t)} = \frac{1}{2 c_i} \tanh \left(\frac{c_i}{2}\right)$
			\EndFor
			\For{$j = 1,2,...,p$}
			\State $\eta_j^{(t)}=
			\lambda + \frac{1}{2}\ln\frac{v_0}{v_1}+
			\frac{1}{2}\left[\left(\mu_j^{(t)}\right)^2+\Sigma_{jj}^{(t)}\right]\left(\frac{1}{v_0}-\frac{1}{v_1}\right)$
			\State $w_j^{(t)}=
			\mathrm{expit}\left(\eta_j^{(t)}\right)$
			\EndFor
			\State$t = t + 1$
			\EndWhile 
			\Ensure $(\boldsymbol{\mu},\mathbf{w},\tau)$ where $\boldsymbol{\mu}\in\mathbb{R}^p,
			\mathbf{w}\in\mathbb{R}^p,\tau > 0$ 
		\end{algorithmic}
	\end{algorithm}
	\section{Conclusion}
	In this paper, we investigate a scalable and interpretable mean-field variational approximation of the spike-and-slab with mixture normal distribution in linear regression, and we promote it to the quantile regression and logistic regression. We derive some theoretical
	results in our approach, (1) explaining
	how to achieve the sparsity estimator;
	(2) proving the estimator's consistency;
	(3) proving the variable selection's consistency.\\
	The method in this paper can be promoted to some more complex models such as
	composite model, mixed model, a model with censored data . And also we can consider some other priors such as 
	horseshoe\cite{horseshoe}, negative-exponential-gamma
	\cite{disscusion1},g priors\cite{gpriors}.\\
	In addition, for the rapid selection of hyperparameters $v_0$ and $v_1$, most of the current research is to use empirical methods\cite{aos2014} to give a specific selection formula or use the grid search method, but they have not given some corresponding theoretical guarantees. How to quickly select a hyperparameter with the theoretically guaranteed $v_0$ and $v_1$ is also a research direction in the future.

	\section*{Appendix 1}
	\subsection*{Appendix 1.1}
	\begin{proof}[\textbf{Lemma $1$ proof}]
		See \cite{Ormerod2017}Appendix B.
	\end{proof}
	\subsection*{Appendix 1.2}
	Before proving the proposition $1,2$, some lemmas are given first. First, some basic conclusions of matrices.
	\begin{lemma}
	If $A, B$ are positive definite matrices of order $p$, then $\operatorname{tr}(AB)>0$.
    $A \leq B \iff B - A$ positive semi-definite,
		but
    When $A \leq B $, there is $B^{-1} \leq A^{-1}$;
    And the diagonal elements satisfy $A_{j,j} \leq B_{j,j},\quad
    1 \leq j \leq p$.
	\end{lemma}
		\begin{lemma}
		$$
		\begin{aligned}
			{\left[\begin{array}{cc}
					\mathbf{A} & \mathbf{B} \\
					\mathbf{B}^T & \mathbf{C}
				\end{array}\right]^{-1} } &=\left[\begin{array}{cc}
				\mathbf{I} & \mathbf{0} \\
				-\mathbf{C}^{-1} \mathbf{B}^T & \mathbf{I}
			\end{array}\right]\left[\begin{array}{cc}
				\tilde{\mathbf{A}} & \mathbf{0} \\
				\mathbf{0} & \mathbf{C}^{-1}
			\end{array}\right]\left[\begin{array}{cc}
				\mathbf{I} & -\mathbf{B C}^{-1} \\
				\mathbf{0} & \mathbf{I}
			\end{array}\right] \\
			&=\left[\begin{array}{cc}
				\tilde{\mathbf{A}} & -\widetilde{\mathbf{A}} \mathbf{B C}^{-1} \\
				-\mathbf{C}^{-1} \mathbf{B}^T \tilde{\mathbf{A}} & \mathbf{C}^{-1}+\mathbf{C}^{-1} \mathbf{B}^T \widetilde{\mathbf{A}} \mathbf{B C}^{-1}
			\end{array}\right]
		\end{aligned}
		$$
		where $\widetilde{\mathbf{A}}=\left(\mathbf{A}-\mathbf{B C}^{-1} \mathbf{B}^T\right)^{-1}.$ 
	\end{lemma}
	\begin{lemma}
		Let $\boldsymbol{\Sigma}$ be the definition of $(5)$, and $\mathbf{D}_0$ be a diagonal matrix whose diagonal elements are greater than $0$, then we have$$ 0<\operatorname{tr}\left(\mathbf{X}^T \mathbf{X} \boldsymbol{\Sigma}\right),0 < \operatorname{tr}\left[
		\mathbf{X}^T\mathbf{X}\left(\mathbf{X}^T\mathbf{X}+\mathbf{D}_0 \right)^{-1} \right] < p.$$
	\end{lemma}
	\begin{proof}[Proof]
		By definition $\tau > 0$ and $\boldsymbol{\Sigma},\mathbf{D}_0 $ is positive definite, and $n \geq p$ and $\mathbf{X}$ is full rank, so also
    $\mathbf{X}^T \mathbf{X}$ is positive definite, then by Lemma $2$
    $\operatorname{tr}\left(\mathbf{X}^T \mathbf{X} \boldsymbol{\Sigma}\right) > 0$.
    Similarly, $$\operatorname{tr}\left[
    \mathbf{X}^T\mathbf{X}\left(\mathbf{X}^T\mathbf{X}+\mathbf{D}_0 \right)^{-1} \right] > 0,\operatorname {tr}\left[
    \mathbf{D}_0
    \left(\mathbf{X}^T\mathbf{X}+
    \mathbf{D}_0 \right)^{-1}
    \right] > 0.$$
    Therefore
		$$
		\begin{aligned}
			\operatorname{tr}\left[
			\mathbf{X}^T\mathbf{X}\left(\mathbf{X}^T\mathbf{X}+\mathbf{D}_0  \right)^{-1} \right]
			&=\operatorname{tr}\left[ \mathbf{I}_p - \mathbf{D}_0 
			\left(\mathbf{X}^T\mathbf{X}+
			\mathbf{D}_0  \right)^{-1}\right]\\
			&=p-\operatorname{tr}\left[
			\mathbf{D}_0 
			\left(\mathbf{X}^T\mathbf{X}+
			\mathbf{D}_0  \right)^{-1}
			\right]\\
			& < p
		\end{aligned}
		$$
	\end{proof}
	Next, we will give the conclusion of the parameter changes in the iterative process, because the least squares estimation is needed
    properties, so the following lemmas assume that the observation matrix $\mathbf{X}$ is full rank, and $n \geq p$.
    \begin{lemma}
		Take $\boldsymbol{\mu}$ as $(5)$,
		then $\left\|\mathbf{X} \boldsymbol{\mu}\right\|^2 \leq 
		\mathbf{y}^T \mathbf{X}\left(\mathbf{X}^T \mathbf{X}\right)^{-1} \mathbf{X}^T \mathbf{y}.$
	\end{lemma}
	\begin{proof}[Proof]
	Because $ \boldsymbol{\mu} = \tau\boldsymbol{\Sigma}\mathbf{X}^T \mathbf{y} = 
	\left(\mathbf{X}^T\mathbf{X}+\mathbf{D}/\tau \right)^{-1}\mathbf{X}^T \mathbf{y}$,
	from lemma 2 ,$$
	\begin{aligned}
		&\left(\mathbf{X}^T\mathbf{X}+\mathbf{D}/\tau \right)^{-1}\mathbf{X}^T\mathbf{X}\left(\mathbf{X}^T\mathbf{X}+\mathbf{D}/\tau \right)^{-1}\\
		=&\left(\mathbf{X}^T\mathbf{X}+\mathbf{D}/\tau \right)^{-1}\left[ \mathbf{I}_p - \left(\mathbf{D}/\tau\right)
		\left(\mathbf{X}^T\mathbf{X}+
		\mathbf{D}/\tau \right)^{-1}\right] \\
		=&\left(\mathbf{X}^T\mathbf{X}+\mathbf{D}/\tau \right)^{-1}
		-\left(\mathbf{X}^T\mathbf{X}+\mathbf{D}/\tau \right)^{-1}
		\left(\mathbf{D}/\tau\right)
		\left(\mathbf{X}^T\mathbf{X}+
		\mathbf{D}/\tau \right)^{-1}\\
		\leq & \left(\mathbf{X}^T\mathbf{X}+\mathbf{D}/\tau \right)^{-1}\\
		\leq & \left(\mathbf{X}^T\mathbf{X}\right)^{-1}
	\end{aligned}
	$$
	So
	$$
	\begin{aligned}
		\left\|\mathbf{X} \boldsymbol{\mu}\right\|^2
		&=\mathbf{y}^T\mathbf{X}
		\left(\mathbf{X}^T\mathbf{X}+\mathbf{D}/\tau \right)^{-1}\mathbf{X}^T\mathbf{X}\left(\mathbf{X}^T\mathbf{X}+\mathbf{D}/\tau \right)^{-1}
		\mathbf{X}^T \mathbf{y}\\
		&\leq \mathbf{y}^T \mathbf{X}\left(\mathbf{X}^T \mathbf{X}\right)^{-1} \mathbf{X}^T \mathbf{y}.
	\end{aligned}
	$$
	\end{proof}
	\begin{lemma}
		There is constant $\tau_R > \tau_L > 0$, that in the iterative process of the algorithm $1$, $\forall t \geq 1,\tau_L \leq \tau^{(t)} \leq \tau_R.$ where
		\begin{equation}
			\begin{gathered}
				\tau_L=\frac{2 A+n-p}{2 B+2\|\mathbf{y}\|^2+2 \mathbf{y}^T \mathbf{X}\left(\mathbf{X}^T \mathbf{X}\right)^{-1} \mathbf{X}^T \mathbf{y}+p \frac{2 A+n-p}{(2 A+n) \tau^{(0)}}} \\
				\tau_R=\frac{2 A+n}{2 B+\left\|\mathbf{y}-\mathbf{X}\left(\mathbf{X}^T \mathbf{X}\right)^{-1} \mathbf{X}^T \mathbf{y}\right\|^2}
			\end{gathered}
		\end{equation}
	\end{lemma}
	\begin{proof}[Proof]
		From algorithm $1$, $$\tau^{(t)} = 
		\frac{2A+n}
		{2B+\left\|\mathbf{y}-\mathbf{X} \boldsymbol{\mu}^{(t)}\right\|^2+\operatorname{tr}\left(\mathbf{X}^T \mathbf{X} \boldsymbol{\Sigma}^{(t)}\right)}$$
		From the properties of least squares estimation
		$$\left\|\mathbf{y}-\mathbf{X} \boldsymbol{\mu}^{(t)}\right\|^2 \geq
		\left\|\mathbf{y}-\mathbf{X}\left(\mathbf{X}^T \mathbf{X}\right)^{-1} \mathbf{X}^T \mathbf{y}\right\|^2
		$$,from lemma $2$,we have $\operatorname{tr}\left(\mathbf{X}^T \mathbf{X} \boldsymbol{\Sigma}^{(t)}\right) > 0$,
		therefore 
		$$
		\tau^{(t)} \leq \frac{2 A+n}{2 B+\left\|\mathbf{y}-\mathbf{X}\left(\mathbf{X}^T \mathbf{X}\right)^{-1} \mathbf{X}^T \mathbf{y}\right\|^2}
		 = \tau_R.
		$$
		From lemma $5$,
		$$
		\left\|\mathbf{X} \boldsymbol{\mu}^{(t)}\right\|^2 
		\leq 
		\mathbf{y}^T \mathbf{X}\left(\mathbf{X}^T \mathbf{X}\right)^{-1} \mathbf{X}^T \mathbf{y}.
		$$
		then the rest proof is the same as \cite{Ormerod2017} Appendix B result 3's proof.
	\end{proof}
	To simplify the proof, some notations are given first, let the index vector
    $\boldsymbol{\gamma} \in \{0,1\}^p$, 
    the related symbol's definition is the same as \cite{Ormerod2017}.
    Now given the indicator vector $\boldsymbol{\gamma}$, consider the matrix
    $\boldsymbol{\Sigma}^{(t)}$ is rewritten as
    \begin{equation}
	{\left[\begin{array}{cc}
			\boldsymbol{\Sigma}_{\boldsymbol{\gamma}, \boldsymbol{\gamma}}^{(t)} & \boldsymbol{\Sigma}_{\boldsymbol{\gamma},-\boldsymbol{\gamma}}^{(t)} \\
			\boldsymbol{\Sigma}_{-\boldsymbol{\gamma}, \boldsymbol{\gamma}}^{(t)} & \boldsymbol{\Sigma}_{-\boldsymbol{\gamma},-\boldsymbol{\gamma}}^{(t)}
		\end{array}\right]} 
		=\left[\begin{array}{cc}
		\tau^{(t)} \mathbf{X}_\gamma^T \mathbf{X} +\mathbf{D}_{\boldsymbol{\gamma}}^{(t)} & \tau^{(t)} \mathbf{X}_{\boldsymbol{\gamma}} \mathbf{X}_{-\boldsymbol{\gamma}} \\
		\tau^{(t)} \mathbf{X}_{-\boldsymbol{\gamma}}^T \mathbf{X}_{\boldsymbol{\gamma}} 
		& \tau^{(t)} \mathbf{X}_{-\boldsymbol{\gamma}}^T \mathbf{X}_{-\boldsymbol{\gamma}}+\mathbf{D}_{-\boldsymbol{\gamma}}^{(t)}
	\end{array}\right]^{-1}
	\end{equation}
	Because $\mathbf{D}^{(t)}$ is a diagonal matrix whose diagonal elements are greater than $0$, and $0 < \tau_L \leq \tau^{(t)} \leq \tau_R$,
    Therefore, by the lemma $2$,
    \begin{equation}
	\Sigma^{(t)}
	\leq \left[\begin{array}{cc}
		\tau_L \mathbf{X}_{\boldsymbol{\gamma}}^T \mathbf{X}_{\boldsymbol{\gamma}} +\mathbf{D}_{\boldsymbol{\gamma}}^{(t)} & \tau_L \mathbf{X}_{\boldsymbol{\gamma}}^T \mathbf{X}_{-\boldsymbol{\gamma}} \\
		\tau_L \mathbf{X}_{-\boldsymbol{\gamma}}^T \mathbf{X}_{\boldsymbol{\gamma}}
		& \tau_L \mathbf{X}_{-\boldsymbol{\gamma}}^T \mathbf{X}_{-\boldsymbol{\gamma}}
	\end{array}\right]^{-1} \stackrel{\mathrm{def}}{=} 
	\boldsymbol{\Omega}^{(t)}
	\end{equation}
	\begin{lemma}
		Take $\boldsymbol{\Sigma}^{(t)}$ as $(14)$,
        let $$s_j =\tau_L \mathbf{X}_{j}^T \mathbf{X}_{j}-\tau_L \mathbf{X}_{j}^T \mathbf{X}_{-j}
        \left(\mathbf{X}_{-j}^T \mathbf{X}_{-j}
        \right)^{-1}\mathbf{X}_{-j}^T \mathbf{X}_{j}$$,
		then $s_j > 0$ and
		\begin{equation}
			\Sigma_{j,j}^{(t)} \leq
			\left[
			\frac{1}{v_0} + \left(\frac{1}{v_1}-\frac{1}{v_0}\right)w_{j}^{(t-1)}
			+s_j\right]^{-1} 
		\end{equation}
	\end{lemma}
		\begin{proof}
		From the formula $(15)$ and the lemma $2$, it is only necessary to prove that $
        \boldsymbol{\Omega}_{j,j}^{(t)}$ is less than or equal to the right side of $(16)$.
		$$
		\begin{aligned}
			\boldsymbol{\Omega}_{\boldsymbol{\gamma},\boldsymbol{\gamma}}^{(t)}
			=& \left[\tau_L \mathbf{X}_{\boldsymbol{\gamma}}^T \mathbf{X}_{\boldsymbol{\gamma}} +\mathbf{D}_{\boldsymbol{\gamma}}^{(t)}
			-\tau_L \mathbf{X}_{\boldsymbol{\gamma}}^T \mathbf{X}_{-\boldsymbol{\gamma}}
			\left(\tau_L \mathbf{X}_{-\boldsymbol{\gamma}}^T \mathbf{X}_{-\boldsymbol{\gamma}}
			\right)^{-1} \tau_L \mathbf{X}_{-\boldsymbol{\gamma}}^T 
			\mathbf{X}_{\boldsymbol{\gamma}}\right]^{-1}\\ 
			=& \left[\tau_L \mathbf{X}_{\boldsymbol{\gamma}}^T \mathbf{X}_{\boldsymbol{\gamma}} +
			\frac{1}{v_0}\mathbf{I_{\boldsymbol{\gamma}}} + \left(\frac{1}{v_1}-\frac{1}{v_0}\right)\mathbf{W}_{\boldsymbol{\gamma}}^{(t-1)}\right.\\
			&\left.-
			\tau_L \mathbf{X}_{\boldsymbol{\gamma}}^T \mathbf{X}_{-\boldsymbol{\gamma}}
			\left(\mathbf{X}_{-\boldsymbol{\gamma}}^T \mathbf{X}_{-\boldsymbol{\gamma}}
			\right)^{-1}  \mathbf{X}_{-\boldsymbol{\gamma}}^T \mathbf{X}_{\boldsymbol{\gamma}}\right]^{-1} 
		\end{aligned}
		$$
		When the index vector $\boldsymbol{\gamma}$ only has the $j$th element non-zero, we get
		\begin{equation}
		\begin{aligned}
			\Omega_{j,j}^{(t)}
			=&
			\left[\tau_L \mathbf{X}_{j}^T \mathbf{X}_{j} +
			\frac{1}{v_0} + \left(\frac{1}{v_1}-\frac{1}{v_0}\right)w_{j}^{(t-1)}
			\right.\\
			&\left.-
			\tau_L \mathbf{X}_{j}^T \mathbf{X}_{-j}
			\left(\mathbf{X}_{-j}^T \mathbf{X}_{-j}
			\right)^{-1}  \mathbf{X}_{-j}^T \mathbf{X}_{j}\right]^{-1} \\
			=&\left[
			\frac{1}{v_0} + \left(\frac{1}{v_1}-\frac{1}{v_0}\right)w_{j}^{(t-1)}
			+s_j\right]^{-1} 
		\end{aligned}
		\end{equation}
		Therefore
		$$
		\Sigma_{j,j}^{(t)} \leq
		\left[
		\frac{1}{v_0} + \left(\frac{1}{v_1}-\frac{1}{v_0}\right)w_{j}^{(t-1)}
		+s_j\right]^{-1} .
		$$
		Similarly, from Lemma $3$, we can see that
        $s_j$ is the diagonal element of the matrix $(\tau_L\mathbf{X}^T \mathbf{X})^{-1}$ according to $(14)$, and $s_j > 0$ can be known from the properties of positive definite matrix
	\end{proof}
	\begin{lemma}
		For $\boldsymbol{\mu}^{(t)}$ in the algorithm $1$, $\forall 1 \leq j \leq p$, there is a constant $c_j > 0$ such that
        $$
        \mu_j^{(t)} \leq c_j\Sigma_{jj}^{(t)}.
        $$
	\end{lemma}
	\begin{proof}
		From $(14)$, using the formula of Lemma $3$, we can know
		$$\boldsymbol{\Sigma}_{\boldsymbol{\gamma},-\boldsymbol{\gamma}}^{(t)}=
		-\boldsymbol{\Sigma}_{\boldsymbol{\gamma}, \boldsymbol{\gamma}}^{(t)}  \mathbf{X}_{\boldsymbol{\gamma}}^T \mathbf{X}_{-\boldsymbol{\gamma}} \left( \mathbf{X}_{-\boldsymbol{\gamma}}^T \mathbf{X}_{-\boldsymbol{\gamma}} + \mathbf{D}_{-\boldsymbol{\gamma}}^{(t)}/\tau^{(t)}
		\right)^{-1}$$
		and
		$$
		\begin{aligned}
		\boldsymbol{\mu}_{\boldsymbol{\gamma}}^{(t)}
		=&\left[\begin{array}{ll}
		\boldsymbol{\Sigma}_{\boldsymbol{\gamma}, \boldsymbol{\gamma}}^{(t)} & \boldsymbol{\Sigma}_{\boldsymbol{\gamma},
		-\boldsymbol{\gamma}}^{(t)}
		\end{array}\right]\left[\begin{array}{c}
		\tau^{(t)}  \mathbf{X}_{\boldsymbol{\gamma}}^T \mathbf{y} \\
		\tau^{(t)}  \mathbf{X}_{-\boldsymbol{\gamma}}^T \mathbf{y}
		\end{array}\right]\\
		=& \tau^{(t)}  \boldsymbol{\Sigma}_{\boldsymbol{\gamma}, \boldsymbol{\gamma}}^{(t)}\left[
		\mathbf{X}_{\boldsymbol{\gamma}}^T \mathbf{y}
		-\mathbf{X}_{\boldsymbol{\gamma}}^T \mathbf{X}_{-\boldsymbol{\gamma}} \left( \mathbf{X}_{-\boldsymbol{\gamma}}^T \mathbf{X}_{-\boldsymbol{\gamma}} +\mathbf{D}_{-\boldsymbol{\gamma}}^{(t)}
				/\tau^{(t)}
				\right)^{-1}
				\mathbf{X}_{-\boldsymbol{\gamma}}^T \mathbf{y}
				\right]
			\end{aligned}
			$$
			Similarly, when the index vector $\boldsymbol{\gamma}$ only has the $j$th element non-zero
			$$
			\mu_{j}^{(t)}=
			\tau^{(t)}  \Sigma_{jj}^{(t)}\left[
			\mathbf{X}_{j}^T \mathbf{y}
			-\mathbf{X}_{j}^T \mathbf{X}_{-j} \left( \mathbf{X}_{-j}^T \mathbf{X}_{-j} +\mathbf{D}_{-j}^{(t)}
			/\tau^{(t)}
			\right)^{-1}
			\mathbf{X}_{-j}^T \mathbf{y}
			\right]
			$$
			By the Cauchy-Schwartz inequality
			$$
			\begin{aligned}
				&\mathbf{X}_{j}^T \mathbf{X}_{-j} \left( \mathbf{X}_{-j}^T \mathbf{X}_{-j} +\mathbf{D}_{-j}^{(t)}
				/\tau^{(t)}
				\right)^{-1}
				\mathbf{X}_{-j}^T \mathbf{y}\\
				\leq& 
				\left[
				\mathbf{X}_{j}^T \mathbf{X}_{-j} \left( \mathbf{X}_{-j}^T \mathbf{X}_{-j} +\mathbf{D}_{-j}^{(t)}
				/\tau^{(t)}
				\right)^{-1}\mathbf{X}_{-j}^T \mathbf{X}_{j}
				\right]^{\frac{1}{2}}\\
				&\times
				\left[
				\mathbf{y}^T\mathbf{X}_{-j} \left( \mathbf{X}_{-j}^T \mathbf{X}_{-j} +\mathbf{D}_{-j}^{(t)}
				/\tau^{(t)}
				\right)^{-1}\mathbf{X}_{-j}^T \mathbf{y}
				\right]^{\frac{1}{2}}\\
				\leq&
				\left[
				\mathbf{X}_{j}^T \mathbf{X}_{-j} \left( \mathbf{X}_{-j}^T \mathbf{X}_{-j}
				\right)^{-1}\mathbf{X}_{-j}^T \mathbf{X}_{j}
				\right]^{\frac{1}{2}}
				\left[
				\mathbf{y}^T\mathbf{X}_{-j} \left( \mathbf{X}_{-j}^T \mathbf{X}_{-j}
				\right)^{-1}\mathbf{X}_{-j}^T \mathbf{y}
				\right]^{\frac{1}{2}} \stackrel{\mathrm{def}}{=} c_0.\\
			\end{aligned}
			$$
			So let $c_j =\tau_R\left(\mathbf{X}_{j}^T \mathbf{y}+c_0\right)$, then $\mu_j^{(t)} \leq c_j\Sigma_{jj}^{(t)}$ .
	\end{proof}
	The following will prove the proposition $1,2$ based on the above lemma.
	\begin{proof}[\textbf{Proposition $1$ Proof}.]
		By $(17)$,$\lim\limits_{w_j^{(t-1)} \to 0} \Omega_{jj}^{(t)} = \left(
		\frac{1}{v_0} + s_j
		\right)^{-1}, $then
		$$
		\lim\limits_{v_0 \to 0}\lim\limits_{w_j^{(t-1)} \to 0} \Omega_{jj}^{(t)} = 0
		$$
		Therefore
		$$
		\lim\limits_{v_0 \to 0}\lim\limits_{w_j^{(t-1)} \to 0} \Sigma_{jj}^{(t)} = 0
		$$
		By Lemma $8$ we get
		$$
		\lim\limits_{v_0 \to 0}\lim\limits_{w_j^{(t-1)} \to 0}
		\mu_j^{(t)} = 0.
		$$
	\end{proof}
	\begin{proof}[\textbf{Proposition $2$ Proof}.]
		When $w_j^{(t-1)} << 1$, because $v_0$ is fixed and $v_1 > v_0 > 0$, so
        $-\frac{1}{v_0} < \frac{1}{v_1}-\frac{1}{v_0} < 0$, and by Lemma $8,s_j > 0$, so
        By Taylor expansion, we get
		$$
		\left[
		\frac{1}{v_0} + \left(\frac{1}{v_1}-\frac{1}{v_0}\right)w_{j}^{(t-1)}
		+s_j\right]^{-1} =
		\left(
		\frac{1}{v_0} + s_j
		\right)^{-1} + \mathrm{O}(w_j^{(t-1)})
		$$
		Let $h_j = \left(\frac{1}{v_0} + s_j\right)^{-1}$,
		Then by Lemma $7$ and Lemma $8$ we get
		$$
		\Sigma_{jj}^{(t)} \leq h_j + \mathrm{O}(w_j^{(t-1)}),
		\quad \mu_j^{(t)} \leq h_j c_j + \mathrm{O}(w_j^{(t-1)})
		$$
		So
		$$
		\Sigma_{jj}^{(t)} + \left(\mu_j^{(t)}\right)^2
		\leq h_j^2c_j^2+h_j+\mathrm{O}(w_j^{(t-1)}).
		$$
		Let $M_j = \lambda + \frac{1}{2v_0}(h_j^2c_j^2+h_j)$,
        From the definition of $h_j, c_j$, it can be seen that $M_j$ is not related to $v_1$,
		$$
		\begin{aligned}
		w_j^{(t)} &= \mathrm{expit}\left[
		\lambda + \frac{1}{2}\ln\frac{v_0}{v_1}+
		\frac{1}{2}\left(\left(\mu_j^{(t)}\right)^2+\Sigma_{jj}^{(t)}\right)\left(\frac{1}{v_0}-\frac{1}{v_1}\right)\right]\\
		& \leq 
		\mathrm{expit}\left[
		\lambda + \frac{1}{2}\ln\frac{v_0}{v_1}+
		\frac{1}{2}\left(\left(\mu_j^{(t)}\right)^2+\Sigma_{jj}^{(t)}\right)\frac{1}{v_0}\right]\\
		&\leq\mathrm{expit}\left[M_{j} + \frac{1}{2}\ln\frac{v_0}{v_1}  + \mathrm{O}(w_{j}^{(t-1)})\right].
		\end{aligned}
		$$
	\end{proof}
	\section*{Appendix 2}
	In this section, we will consider the proof of theorem $1,2$ under the assumption $1-5$,
    Because it is assumed that $\mathbf{X},\mathbf{y}$ are all random variables, first define some random sequences
    $$
	\begin{gathered}
		\mathbf{A}_n=n^{-1} \mathbf{X}^T \mathbf{X}, \quad \mathbf{b}_n=n^{-1} \mathbf{X}^T \mathbf{y} \\
		c_n=\operatorname{tr}\left[\mathbf{A}_n\left(\mathbf{A}_n+ n^{-1} \tau^{-1} \mathbf{D}\right)^{-1}\right], \quad  \boldsymbol{\beta}_{\mathrm{LS}}=\left(\mathbf{X}^T \mathbf{X}\right)^{-1} \mathbf{X}^T \mathbf{y}
	\end{gathered}
	$$
	From lemma 4-6 in \cite{Ormerod2022}
    , we can get
	\begin{equation}
		\begin{gathered}
			\mathbf{A}_n 
			\stackrel{\text { P }}{\to}
			\mathbf{S}\quad,\quad
			\mathbf{A}_n^{-1}\stackrel{\text { P }}{\to}\mathbf{S}^{-1},\\ 
			n^{-1}\|\boldsymbol{\epsilon}\|^2 \stackrel{\text { P }}{\to} \sigma_0^2\quad,\quad
			n^{-1} \mathbf{X}^T \boldsymbol{\epsilon}
			\stackrel{\text { P }}{\to} \mathbf{0},\\
			\mathbf{b}_n =n^{-1} \mathbf{X}^T\left(\mathbf{X} \boldsymbol{\beta}_0+\boldsymbol{\epsilon}\right)
			\stackrel{\text { P }}{\to}
			\mathbf{S} \boldsymbol{\beta}_0.
		\end{gathered}
	\end{equation}
	In order to get the consistency estimate of $\boldsymbol{\mu}$,
    First prove that $\forall t \geq 1$, $\tau^{(t)}$ is bounded by probability.
    \begin{lemma}
		Under the condition of assumption $1-4$, 
		$$\tau^{(t)} = \mathrm{O}_p(1),\quad1/\tau^{(t)} = \mathrm{O}_p(1) \quad (\forall t \geq 0)$$
	\end{lemma}	
	The proof is in \cite{Ormerod2017} Appendix B result 5.
	\begin{proof}[\textbf{Theorem $1$ proof}.]
		We will first prove the consistency of $\boldsymbol{\mu}^{(t)}$,
		because$$
		\begin{aligned}
			\boldsymbol{\mu}^{(t)}
			=& 
			\left[
			\tau^{(t-1)}\mathbf{X}^T\mathbf{X}+\mathbf{D}^{(t)}
			\right]^{-1}\mathbf{X}^T \mathbf{y}\\
			=&\left[
			\frac{1}{n}\mathbf{X}^T\mathbf{X}+
			\frac{1}{n\tau^{(t-1)}}\mathbf{D}^{(t)}
			\right]^{-1}\times\frac{1}{n}\mathbf{X}^T \mathbf{y}\\
			=&\left[
			\frac{1}{n}\mathbf{X}^T\mathbf{X}+
			\frac{1}{n\tau^{(t-1)}}\mathbf{D}^{(t)}
			\right]^{-1}\mathbf{b}_n
		\end{aligned}
		 $$from $(31)$ ,$\frac{1}{n}\mathbf{X}^T\mathbf{X} \stackrel{\mathrm{P}}{\to} \mathbf{S},
		 \mathbf{b}_n \stackrel{\text { P }}{\to}
		 \mathbf{S} \boldsymbol{\beta}_0$,but the lemma
		 $10,1/\tau^{(t-1)} = \mathrm{O}_p(1)$,
		 therefore $$
		 \boldsymbol{\mu}^{(t)} \stackrel{\text { P }}{\to} 
		 \mathbf{S}^{-1}\mathbf{S}\boldsymbol{\beta}_0 = \boldsymbol{\beta}_0.
		 $$
		 we will next prove the consistency of $1/\tau^{(t)}$,
		 $$
		 \begin{aligned}
		 	1/\tau^{(t)} &= 
		 	\frac
		 	{2B+\left\|\mathbf{y}-\mathbf{X} \boldsymbol{\mu}^{(t)}\right\|^2+\operatorname{tr}\left(\mathbf{X}^T \mathbf{X} \boldsymbol{\Sigma}^{(t)}\right)}{2A+n}\\
		 	&= \frac
		 	{2B/n+n^{-1}\left\|\mathbf{y}-\mathbf{X} \boldsymbol{\mu}^{(t)}\right\|^2+n^{-1}\operatorname{tr}\left(\mathbf{X}^T \mathbf{X} \boldsymbol{\Sigma}^{(t)}\right)}{2A/n+1}
		 \end{aligned}
		 $$
		 From lemma 6, 
		 $\operatorname{tr}\left(\mathbf{X}^T \mathbf{X} \boldsymbol{\Sigma}^{(t)} \right)
		 \leq p / \tau^{(t-1)}$,from lemma 9
		 $$n^{-1}\operatorname{tr}\left(\mathbf{X}^T \mathbf{X} \boldsymbol{\Sigma}^{(t)}\right) \stackrel{\mathrm{P}}{\to} 0.$$
		 And$\mathbf{y} = \mathbf{X}\boldsymbol{\beta}_0 + \boldsymbol{\epsilon}$,so
		 $$
		 \begin{aligned}
		 	n^{-1}\left\|\mathbf{y}-\mathbf{X} \boldsymbol{\mu}^{(t)}\right\|^2
		 	=&n^{-1}\|\boldsymbol{\epsilon}\|^2+
		 	2\left(\boldsymbol{\beta}_0-
		 	\boldsymbol{\mu}^{(t)}\right)^T
		 	\left(n^{-1}\mathbf{X}^T \boldsymbol{\epsilon}\right) \\
		 	&+\left(\boldsymbol{\beta}_0-\boldsymbol{\mu}^{(t)}\right)^T\left(n^{-1} \mathbf{X}^T \mathbf{X}\right)
		 	\left(\boldsymbol{\beta}_0-\boldsymbol{\mu}^{(t)}\right)
		 \end{aligned}
		 $$
		 From $(31)$,$n^{-1}\|\boldsymbol{\epsilon}\|^2 \stackrel{\text { P }}{\to} \sigma_0^2,n^{-1} \mathbf{X}^T \boldsymbol{\epsilon}
		 \stackrel{\text { P }}{\to} \mathbf{0},
		 n^{-1} \mathbf{X}^T \mathbf{X} \stackrel{\text { P }}{\to} \mathbf{S}$.
		 And $\mathbf{S}$ has finite elements, and$\boldsymbol{\mu}^{(t)} \stackrel{\text { P }}{\to} \boldsymbol{\beta}_0,$
		 So $1/\tau^{(t)} \stackrel{\text { P }}{\to} \sigma_0^2$.
	\end{proof}
	Because in the iterative process of algorithm $1$, it contains probability $w_j^{(t)}$ and covariance matrix
    $\boldsymbol{\Sigma}^{(t)}$ is related, so before proving theorem $2$, we will give
    Asymptotic properties of $\boldsymbol{\Sigma}^{(t)}$.
    \begin{lemma}
		Under the condition of assumption $1-4$, $\forall t \geq 1$,
		$$\boldsymbol{\Sigma}^{(t)} \stackrel{\mathrm{P}}{\to} \mathbf{0} \quad (n \to +\infty)$$
	\end{lemma}
	\begin{proof}[proof]
		From the definition
		$$
		\begin{aligned}
			\boldsymbol{\Sigma}^{(t)} = \left[
			\tau^{(t-1)}\mathbf{X}^T\mathbf{X}+\mathbf{D}^{(t)}\right]^{-1}
			= \frac{1}{n}\left[
			\frac{1}{n}\mathbf{X}^T\mathbf{X}+
			\frac{1}{n\tau^{(t-1)}}\mathbf{D}^{(t)}
			\right]^{-1}
		\end{aligned}
		$$
		From$(31)$, $\frac{1}{n}\mathbf{X}^T\mathbf{X} \stackrel{\mathrm{P}}{\to} \mathbf{S}$,by the lemma
		$9,1/\tau^{(t-1)} = \mathrm{O}_p(1)$,while the matrix
        $\mathbf{S},\mathbf{D}^{(t)}$'s  element is finite, so
		$\boldsymbol{\Sigma}^{(t)} \stackrel{\mathrm{P}}{\to} \mathbf{0}. \quad (n \to +\infty)$
	\end{proof}
	Next, we will prove theorem $2.$
	\begin{proof}[\textbf{Theorem $2$ proof}.]
		From the definition,we only need to prove $n \to +\infty$,
		$$
		\begin{aligned}
		&(i)\forall j \in S^{\star},\operatorname{P}\left( 
		w_j^{(t)} > 0.5\right) \to 1;\\
		&(ii)\forall j \notin S^{\star},\operatorname{P}\left( 
		w_j^{(t)} \leq 0.5\right) \to 1.
		\end{aligned}
		$$
		We will first prove $(ii)$.$\forall j \notin S^{\star}$, then
		$\boldsymbol{\beta}_{0,j} = 0$,by the theorem $1$,
		$\mu_{j}^{(t)} \stackrel{\mathrm{P}}{\to} 0.$
		Suppose $\eta_j^{(t)}$ is defined in algorithm $1$, then by assumption $5$ and theorem $1$,
		$$
		\begin{aligned}
		\eta_j^{(t)}=&
		\lambda + \frac{1}{2}\ln\frac{v_0}{v_1}+
		\frac{1}{2}\left[\left(\mu_j^{(t)}\right)^2+\Sigma_{jj}^{(t)}\right]\left(\frac{1}{v_0}-\frac{1}{v_1}\right)\\
		\leq & \lambda + \frac{1}{2}\ln\frac{v_1\exp\left(-2\lambda\right)}{v_1}+
		\frac{1}{2}\left[\left(\mu_j^{(t)}\right)^2+\Sigma_{jj}^{(t)}\right]\left(\frac{1}{v_0}-\frac{1}{v_1}\right)\\
		\leq & \frac{1}{2}\left[\left(\mu_j^{(t)}\right)^2+\Sigma_{jj}^{(t)}\right]\left(\frac{1}{v_0}-\frac{1}{v_1}\right) \stackrel{\mathrm{P}}{\to} 0\quad
		(n \to +\infty)
		\end{aligned}
		$$
		Then $\operatorname{P}\left(\eta_j^{(t)} \leq 0\right) \stackrel{\mathrm{P}}{\to} 1$, so
		$$
		\begin{aligned}
		\operatorname{P}\left( 
		w_j^{(t)} \leq 0.5\right) 
		&= 
		\operatorname{P}\left( 
		\mathrm{expit}(\eta_j^{(t)})
		\leq 0.5\right)\\
		&= 
		\operatorname{P}\left( 
		\eta_j^{(t)}\leq 0\right){\to}1\quad
		(n \to +\infty)
		\end{aligned}
		$$
		Therefore, $(ii)$ is proved, and we will consider the proof of $(i)$.\\
		$\forall j \in S^{\star}$,by the assumption $5$, then
		$|\boldsymbol{\beta}_{0,j}| \geq l_0$,
		$\frac{l_0^2}{v_0} + \ln{v_0}
		\geq \ln v_1+\frac{l_0^2}{v_1}-2\lambda + 2\delta$, so
		$$
		\begin{aligned}
			\eta_j^{(t)}&=
			\lambda + \frac{1}{2}\ln\frac{v_0}{v_1}+
			\frac{1}{2}\left[\left(\mu_j^{(t)}\right)^2+\Sigma_{jj}^{(t)}\right]
			\left(\frac{1}{v_0}-\frac{1}{v_1}\right)\\
			&\stackrel{\mathrm{P}}{\to}  \lambda + \frac{1}{2}\ln\frac{v_0}{v_1}+
			\frac{1}{2}\boldsymbol{\beta}_{0,j}^2\left(\frac{1}{v_0}-\frac{1}{v_1}\right)\\
			&\geq  \lambda + \frac{1}{2}\ln\frac{v_0}{v_1}+
			\frac{1}{2}l_0^2\left(\frac{1}{v_0}-\frac{1}{v_1}
			\right)\\
			&\geq \delta
			 > 0
		\end{aligned}
		$$
		Henceforth
		$$
		\begin{aligned}
			\operatorname{P}\left( 
			w_j^{(t)} > 0.5\right) 
			&= 
			\operatorname{P}\left( 
			\mathrm{expit}(\eta_j^{(t)})
			> 0.5\right)\\
			&= 
			\operatorname{P}\left( 
			\eta_j^{(t)}> 0\right){\to}1\quad
			(n \to +\infty)
		\end{aligned}
		$$
		Therefore
		$$\operatorname{P}\left(\hat{S}_n = 
		S^\star\right) \to 1 \quad (n \to +\infty)$$
	\end{proof}
    \section*{Appendix 3}
    \begin{lemma}
		$$
		\begin{aligned}
			& \mathbf{y}^{\top} \mathbf{X} \left(\mathbf{X}^{\top} \mathbf{X}+\frac{1}{v_0}\mathbf{I}+\left(\frac{1}{v_1}-\frac{1}{v_0}\right)\mathbf{W}^{(j1)}\right)^{-1} \mathbf{X}^{\top} \mathbf{y} \\
			& -\mathbf{y}^{\top} \mathbf{X} \left(\mathbf{X}^{\top} \mathbf{X}+\frac{1}{v_0}\mathbf{I}+\left(\frac{1}{v_1}-\frac{1}{v_0}\right)\mathbf{W}^{(j0)}\right)^{-1} \mathbf{X}^{\top} \mathbf{y} \\
			& = \left(\frac{1}{v_0}-\frac{1}{v_1}\right)
			\boldsymbol{\mu}_j^{(j0)}\boldsymbol{\mu}_j^{(j1)}
		\end{aligned}
		$$
	\end{lemma}
    \begin{proof}
    Consider the general situation first, let $\boldsymbol{\Sigma} = \left(\mathbf{X}^{\top} \mathbf{X}+\frac{1}{v_0}\mathbf{I}+\left( \frac{1}{v_1}-\frac{1}{v_0}\right)\mathbf{W}\right)^{-1},\mathbf{D} = \frac{1}{v_0}\mathbf {I}+\left(\frac{1}{v_1}-\frac{1}{v_0}\right)\mathbf{W}$, also consider rewriting $\boldsymbol{\Sigma}$ as\begin{equation}
			\nonumber
			{\left[\begin{array}{cc}
					\Sigma_{j,j} & \boldsymbol{\Sigma}_{j,-j} \\
					\boldsymbol{\Sigma}_{-j, j} & \boldsymbol{\Sigma}_{-j,-j}
				\end{array}\right]} 
			=\left[\begin{array}{cc}
				\mathbf{X}_j^T \mathbf{X}_j + \frac{1}{v_0}+\left(\frac{1}{v_1}-\frac{1}{v_0}\right)w_j
				&\mathbf{X}_{j}^T \mathbf{X}_{-j} \\
				\mathbf{X}_{-j}^T \mathbf{X}_{j} 
				& \mathbf{X}_{-j}^T \mathbf{X}_{-j}+
				\mathbf{D}_{-j}
			\end{array}\right]^{-1}
	\end{equation}
    Let $\mathbf{C}_j = \left(\mathbf{X}_{-j}^T \mathbf{X}_{-j}+\mathbf{D}_{-j}\right)^{ -1}$, from Lemma $3$ we get
    \begin{equation}
		\begin{aligned}
			&\Sigma_{j,j} =\left[
			\mathbf{X}_j^T \mathbf{X}_j + \frac{1}{v_0}+\left(\frac{1}{v_1}-\frac{1}{v_0}\right)w_j - 
			\mathbf{X}_{j}^T \mathbf{X}_{-j}\mathbf{C}_j\mathbf{X}_{-j}^T \mathbf{X}_{j}
			\right]^{-1} \\
			&\boldsymbol{\Sigma}_{j,-j} = 
			-\Sigma_{j,j}\mathbf{X}_{j}^T \mathbf{X}_{-j}\mathbf{C}_j
			\\
			&\boldsymbol{\Sigma}_{-j,j} = 
			-\Sigma_{j,j}\mathbf{C}_j\mathbf{X}_{-j}^T \mathbf{X}_{j}\\
			&\boldsymbol{\Sigma}_{-j,-j} = 
			\left(
			\mathbf{X}_{-j}^T \mathbf{X}_{-j}+
			\mathbf{D}_{-j}
			\right)^{-1} + \Sigma_{j,j}\mathbf{C}_j
			\mathbf{X}_{-j}^T \mathbf{X}_{j}
			\mathbf{X}_{j}^T \mathbf{X}_{-j}
			\mathbf{C}_j
		\end{aligned}
	\end{equation}
    Set $k \in \{0,1\}, let $ $\boldsymbol{\Sigma}^{(jk)} = \left(\mathbf{X}^{\top} \mathbf{X}+\frac {1}{v_0}\mathbf{I}+\left(\frac{1}{v_1}-\frac{1}{v_0}\right)\mathbf{W}^{(jk)}\right)^ {-1}$, replace $w_j$ with $k$, get from (45),
    $$\Sigma_{j,j}^{(jk)} = \left[
	\mathbf{X}_j^T \mathbf{X}_j + \frac{1}{v_k}- 
	\mathbf{X}_{j}^T \mathbf{X}_{-j}\mathbf{C}_j\mathbf{X}_{-j}^T \mathbf{X}_{j}
	\right]^{-1}$$
    Let $d_j = \Sigma_{j,j}^{(j1)} - \Sigma_{j,j}^{(j0)}$, then $d_j = \left(\frac{1}{v_0} -\frac{1}{v_1}\right)\Sigma_{j,j}^{(j1)} \Sigma_{j,j}^{(j0)}$, and
    $$
	\begin{aligned}
		\boldsymbol{\Sigma}^{(j1)} - \boldsymbol{\Sigma}^{(j0)}
		&= {\left[\begin{array}{cc}
				d_j & -d_j\mathbf{X}_{j}^T \mathbf{X}_{-j}\mathbf{C}_j \\
				-d_j\mathbf{C}_j\mathbf{X}_{-j}^T \mathbf{X}_{j} & d_j\mathbf{C}_j
				\mathbf{X}_{-j}^T \mathbf{X}_{j}
				\mathbf{X}_{j}^T \mathbf{X}_{-j}
				\mathbf{C}_j
		\end{array}\right]} 
	\end{aligned}
	$$
    Empathy rewrite$\mathbf{X}^{\top} \mathbf{y} = \left[\begin{array}{c}
	\mathbf{X}_{j}^T \mathbf{y} \\
	\mathbf{X}_{-j}^T \mathbf{y}
	\end{array}\right]$,so 
    $$
	\begin{aligned}
		&\mathbf{y}^{\top} \mathbf{X}
		\left(\boldsymbol{\Sigma}^{(j1)} - \boldsymbol{\Sigma}^{(j0)}\right)
		\mathbf{X}^{\top} \mathbf{y} \\
		&=
		d_j\left[
		\left(\mathbf{X}_{j}^T \mathbf{y}\right)^2
		-2\left(\mathbf{X}_{j}^T \mathbf{y}\right)\left(\mathbf{X}_{j}^T \mathbf{X}_{-j}\mathbf{C}_j\mathbf{X}_{-j}^T \mathbf{y}\right)
		+ \left(\mathbf{X}_{j}^T \mathbf{X}_{-j}\mathbf{C}_j\mathbf{X}_{-j}^T \mathbf{y}\right)^2
		 \right]\\
		\quad &=\left(\frac{1}{v_0}-\frac{1}{v_1}\right)\Sigma_{j,j}^{(j1)} \Sigma_{j,j}^{(j0)}\left(
		\mathbf{X}_{j}^T \mathbf{y} - \mathbf{X}_{j}^T \mathbf{X}_{-j}\left(
		\mathbf{X}_{-j}^T \mathbf{X}_{-j}+
		\mathbf{D}_{-j}
		\right)^{-1}\mathbf{X}_{-j}^T \mathbf{y}
		\right)^2
	\end{aligned}
	$$
    So
 $$
	\mathbf{y}^{\top} \mathbf{X}
	\left(\boldsymbol{\Sigma}^{(j1)} - \boldsymbol{\Sigma}^{(j0)}\right)
	\mathbf{X}^{\top} \mathbf{y} = \left(\frac{1}{v_0}-\frac{1}{v_1}\right)
	\boldsymbol{\mu}_j^{(j0)}\boldsymbol{\mu}_j^{(j1)}.
	$$
    \end{proof}
    \begin{proof}[Theorem $3$ proof]]
    Firstly, it is proved by theorem $1$ and the formula $(44)$, because the assumptions are the same, it can be obtained in the same way
    \begin{equation}
			\label{mujk}
		\boldsymbol{\mu}^{(jk)} = \boldsymbol{\beta}_0 +  \mathbf{O}_p^v
		\left(p^3n^{-\frac{1}{2}}\right),k \in \{0,1\}
		\end{equation}
    and 
    \begin{equation}
			\begin{aligned}
				T_{j 0}-T_{j 1}=&-\left(A+\frac{n}{2}\right) \log \left(1+\frac{\left(\frac{1}{v_0}-\frac{1}{v_1}\right)
				\boldsymbol{\mu}_j^{(j0)}\boldsymbol{\mu}_j^{(j1)}}{B+\frac{1}{2}\|\mathbf{y}\|^2-\frac{1}{2} \mathbf{y}^{\top} \mathbf{X} \boldsymbol{\mu}^{(j1)}}\right) \\
				& -\operatorname{logit}(\rho)+\frac{1}{2} \log n+\frac{1}{2} \log v_1 \\
				=&-\left(A+\frac{n}{2}\right) \log \left(1+c_j\right)-\operatorname{logit}(\rho)+\frac{1}{2} \log n+\frac{1}{2} \log v_1\\
			\end{aligned}
		\end{equation}
    From $\mathbf{y} = \mathbf{X}\boldsymbol{ \beta }_0+\boldsymbol{\epsilon}$ and (\ref{mujk}) can get $B+\frac{1}{2}\| \mathbf{y}\|^2-\frac{1}{2} \mathbf{y}^{\top} \mathbf{X} \boldsymbol{\mu}^{(j1)}= \mathrm{O }_p(n)$, and because of $v_0=\mathrm{O}(n^{\frac{1}{2}}), $if $j \in \boldsymbol{\gamma}_0$, $c_j = \mathrm{O}_p(n^{-\frac{1}{2}})$, $$\left(A+\frac{n}{ 2}\right) \log \left(1+c_j\right) =\mathrm{O}_p(n^{\frac{1}{2}}), $$ at this time $T_{j 0}-T_ {j 1} =-\left|\mathrm{O}_p(n^{\frac{1}{2}})\right|$, so $w_j = 1-\left(1+\exp \left| \mathrm{O}_p(n^{\frac{1}{2}})\right|\right)^{-1}.$\\
If $j \notin \boldsymbol{\gamma}_0$, then
$c_j = \mathrm{O}_P(p^6n^{-\frac{3}{2}})$,
And because $p = \mathrm{O}(n^{\frac{1}{12}})$,
In the same way, $$\left(A+\frac{n}{2}\right) \log \left(1+c_j\right) =
\mathrm{O}_p(p^6n^{-\frac{1}{2}})=\mathrm{O}_p(1),$$
So at this time $T_{j 0}-T_{j 1} =\mathrm{O}_p(\log n)$,
So $w_j = \mathrm{O}_p(n^{-1}).$
    \end{proof}
    \printbibliography{}

@article{Ormerod2017,
  title={A variational Bayes approach to variable selection},
  author={Ormerod, John T and You, Chong and M{\"u}ller, Samuel},
  journal={Electronic Journal of Statistics},
  volume={11},
  number={2},
  pages={3549--3594},
  year={2017},
  publisher={Institute of Mathematical Statistics and Bernoulli Society}
}

@article{EMVS,
  title={EMVS: The EM approach to Bayesian variable selection},
  author={Ro{\v{c}}kov{\'a}, Veronika and George, Edward I},
  journal={Journal of the American Statistical Association},
  volume={109},
  number={506},
  pages={828--846},
  year={2014},
  publisher={Taylor \& Francis}
}

@phdthesis{UIUC2016,
  title={Scalable algorithms for Bayesian variable selection},
  author={Wang, Jin},
  year={2016},
  school={University of Illinois at Urbana-Champaign}
}

@article{2014AUS,
  title={On variational Bayes estimation and variational information criteria for linear regression models},
  author={You, Chong and Ormerod, John T and Mueller, Samuel},
  journal={Australian \& New Zealand Journal of Statistics},
  volume={56},
  number={1},
  pages={73--87},
  year={2014},
  publisher={Wiley Online Library}
}

@article{PG,
  title={Bayesian inference for logistic models using P{\'o}lya--Gamma latent variables},
  author={Polson, Nicholas G and Scott, James G and Windle, Jesse},
  journal={Journal of the American statistical Association},
  volume={108},
  number={504},
  pages={1339--1349},
  year={2013},
  publisher={Taylor \& Francis}
}

@article{disscusion1,
  title={Penalized wavelets: Embedding wavelets into semiparametric regression},
  author={Wand, MP and Ormerod, John T},
  journal={Electronic Journal of Statistics},
  volume={5},
  pages={1654--1717},
  year={2011},
  publisher={Institute of Mathematical Statistics and Bernoulli Society}
}

@article{horseshoe,
  title={Sparsity information and regularization in the horseshoe and other shrinkage priors},
  author={Piironen, Juho and Vehtari, Aki},
  journal={Electronic Journal of Statistics},
  volume={11},
  number={2},
  pages={5018--5051},
  year={2017},
  publisher={Institute of Mathematical Statistics and Bernoulli Society}
}

@article{aos2014,
  title={Bayesian variable selection with shrinking and diffusing priors},
  author={Narisetty, Naveen Naidu and He, Xuming},
  journal={The Annals of Statistics},
  volume={42},
  number={2},
  pages={789--817},
  year={2014},
  publisher={Institute of Mathematical Statistics}
}

@book{Spring,
  title={Statistics for high-dimensional data: methods, theory and applications},
  author={B{\"u}hlmann, Peter and Van De Geer, Sara},
  year={2011},
  publisher={Springer Science \& Business Media}
}

@article{GM97,
  title={Approaches for Bayesian variable selection},
  author={George, Edward I and McCulloch, Robert E},
  journal={Statistica sinica},
  pages={339--373},
  year={1997},
  publisher={JSTOR}
}

@article{lasso,
  title={Regression shrinkage and selection via the lasso},
  author={Tibshirani, Robert},
  journal={Journal of the Royal Statistical Society: Series B (Methodological)},
  volume={58},
  number={1},
  pages={267--288},
  year={1996},
  publisher={Wiley Online Library}
}

@article{SCAD,
  title={Variable selection via nonconcave penalized likelihood and its oracle properties},
  author={Fan, Jianqing and Li, Runze},
  journal={Journal of the American statistical Association},
  volume={96},
  number={456},
  pages={1348--1360},
  year={2001},
  publisher={Taylor \& Francis}
}

@article{AIC,
  title={AIC model selection using Akaike weights},
  author={Wagenmakers, Eric-Jan and Farrell, Simon},
  journal={Psychonomic bulletin \& review},
  volume={11},
  number={1},
  pages={192--196},
  year={2004},
  publisher={Springer}
}

@article{BIC,
  title={Multimodel inference: understanding AIC and BIC in model selection},
  author={Burnham, Kenneth P and Anderson, David R},
  journal={Sociological methods \& research},
  volume={33},
  number={2},
  pages={261--304},
  year={2004},
  publisher={Sage Publications Sage CA: Thousand Oaks, CA}
}

@article{review1,
  title={A review of Bayesian variable selection methods: what, how and which},
  author={O'Hara, Robert B and Sillanp{\"a}{\"a}, Mikko J},
  journal={Bayesian analysis},
  volume={4},
  number={1},
  pages={85--117},
  year={2009},
  publisher={International Society for Bayesian Analysis}
}

@article{gpriors,
  title={Mixtures of g priors for Bayesian variable selection},
  author={Liang, Feng and Paulo, Rui and Molina, German and Clyde, Merlise A and Berger, Jim O},
  journal={Journal of the American Statistical Association},
  volume={103},
  number={481},
  pages={410--423},
  year={2008},
  publisher={Taylor \& Francis}
}

@article{2002,
  title={On Bayesian model and variable selection using MCMC},
  author={Dellaportas, Petros and Forster, Jonathan J and Ntzoufras, Ioannis},
  journal={Statistics and Computing},
  volume={12},
  number={1},
  pages={27--36},
  year={2002},
  publisher={Springer}
}

@article{2000,
  title={Bayesian variable selection using the Gibbs sampler},
  author={Dellaportas, Petros and Forster, Jonathan J and Ntzoufras, Ioannis},
  journal={Biostatistics-Basel-},
  volume={5},
  pages={273--286},
  year={2000},
  publisher={Marcel Dekker Inc}
}

@article{fox2012tutorial,
  title={A tutorial on variational Bayesian inference},
  author={Fox, Charles W and Roberts, Stephen J},
  journal={Artificial intelligence review},
  volume={38},
  number={2},
  pages={85--95},
  year={2012},
  publisher={Springer}
}

@article{explaining,
  title={Explaining variational approximations},
  author={Ormerod, John T and Wand, Matt P},
  journal={The American Statistician},
  volume={64},
  number={2},
  pages={140--153},
  year={2010},
  publisher={Taylor \& Francis}
}

@article{mean,
  title={Mean field variational Bayesian inference for nonparametric regression with measurement error},
  author={Pham, Tung H and Ormerod, John T and Wand, Matthew P},
  journal={Computational Statistics \& Data Analysis},
  volume={68},
  pages={375--387},
  year={2013},
  publisher={Elsevier}
}

@article{logit1,
  title={Spike and slab variational Bayes for high dimensional logistic regression},
  author={Ray, Kolyan and Szab{\'o}, Botond and Clara, Gabriel},
  journal={Advances in Neural Information Processing Systems},
  volume={33},
  pages={14423--14434},
  year={2020}
}

@article{logit2,
  title={A novel variational Bayesian method for variable selection in logistic regression models},
  author={Zhang, Chun-Xia and Xu, Shuang and Zhang, Jiang-She},
  journal={Computational Statistics \& Data Analysis},
  volume={133},
  pages={1--19},
  year={2019},
  publisher={Elsevier}
}

@article{sparse,
  title={Variational Bayes for high-dimensional linear regression with sparse priors},
  author={Ray, Kolyan and Szab{\'o}, Botond},
  journal={Journal of the American Statistical Association},
  volume={117},
  number={539},
  pages={1270--1281},
  year={2022},
  publisher={Taylor \& Francis}
}

@article{frequentist,
  title={Frequentist consistency of variational Bayes},
  author={Wang, Yixin and Blei, David M},
  journal={Journal of the American Statistical Association},
  volume={114},
  number={527},
  pages={1147--1161},
  year={2019},
  publisher={Taylor \& Francis}
}

@article{yang2020alpha,
  title={alpha -variational inference with statistical guarantees},
  author={Yang, Yun and Pati, Debdeep and Bhattacharya, Anirban},
  journal={The Annals of Statistics},
  volume={48},
  number={2},
  pages={886--905},
  year={2020},
  publisher={Institute of Mathematical Statistics}
}

@article{zhang2022bayesian,
  title={Bayesian regression using a prior on the model fit: The R2-D2 shrinkage prior},
  author={Zhang, Yan Dora and Naughton, Brian P and Bondell, Howard D and Reich, Brian J},
  journal={Journal of the American Statistical Association},
  volume={117},
  number={538},
  pages={862--874},
  year={2022},
  publisher={Taylor \& Francis}
}

@article{han2019statistical,
  title={Statistical inference in mean-field variational Bayes},
  author={Han, Wei and Yang, Yun},
  journal={arXiv preprint arXiv:1911.01525},
  year={2019}
}

@article{huang2016variational,
  title={A variational algorithm for Bayesian variable selection},
  author={Huang, Xichen and Wang, Jin and Liang, Feng},
  journal={arXiv preprint arXiv:1602.07640},
  year={2016}
}

@article{carbonetto2012scalable,
  title={Scalable variational inference for Bayesian variable selection in regression, and its accuracy in genetic association studies},
  author={Carbonetto, Peter and Stephens, Matthew},
  journal={Bayesian analysis},
  volume={7},
  number={1},
  pages={73--108},
  year={2012},
  publisher={International Society for Bayesian Analysis}
}

@article{wang2020simple,
  title={A simple new approach to variable selection in regression, with application to genetic fine-mapping},
  author={Wang, Gao and Sarkar, Abhishek and Carbonetto, Peter and Stephens, Matthew},
  journal={BioRxiv},
  pages={501114},
  year={2020},
  publisher={Cold Spring Harbor Laboratory}
}

@article{Ormerod2022,
  title={An Approximated Collapsed Variational Bayes Approach to Variable Selection in Linear Regression},
  author={You, Chong and Ormerod, John T and Li, Xiangyang and Pang, Cheng Heng and Zhou, Xiao-Hua},
  journal={Journal of Computational and Graphical Statistics},
  pages={1--11},
  year={2023},
  publisher={Taylor \& Francis}
}
\end{document}